\documentclass[]{article}

\usepackage{authblk,fancyhdr,amsthm}
\usepackage[letterpaper,margin=1in]{geometry}
\usepackage{xifthen}
\newboolean{camera_ready}
\setboolean{camera_ready}{false}

\newcommand{\therunningauthor}{}
\newcommand{\therunningtitle}{}
\newcommand{\theinstitute}{}

\renewcommand{\and}{\unskip,\xspace}
\newcommand{\authorrunning}[1]{\renewcommand\therunningauthor{#1}}
\newcommand{\titlerunning}[1]{\renewcommand\therunningtitle{#1}}
\newcommand{\institute}[1]{\renewcommand\theinstitute{#1}}
\newcommand{\inst}[1]{$^{\text{#1}}$}
\newcommand{\orcidID}[1]{}
\newcommand{\email}[1]{\url{#1}}
\newcommand{\keywords}[1]{\newcommand{\and}{\unskip,\xspace}\textbf{Keywords:} #1}

\newtheorem{definition}{\textbf{Definition}}

\newtheorem{lemma}{\textbf{Lemma}}

\let\svthefootnote\thefootnote
\newcommand\blankfootnote[1]{%
  \let\thefootnote\relax\footnotetext{#1}%
  \let\thefootnote\svthefootnote%
}
\let\oldmaketitle\maketitle
\renewcommand{\maketitle}{
  \oldmaketitle
  \blankfootnote{$^\text{1}$\theinstitute}
  \pagestyle{fancy}\rhead{\therunningauthor}\lhead{\therunningtitle}
}

\usepackage{graphicx,wrapfig,hyperref,xcolor,xspace}
\usepackage{amsmath,amssymb,mathtools,siunitx}
\usepackage{thmtools,thm-restate,booktabs}

\hypersetup{
    colorlinks,
    linkcolor={red!30!black},
    citecolor={blue!30!black},
    urlcolor={blue!20!black}
}

\newcommand{\kmer}{$k$-mer\xspace}
\newcommand{\kmers}{\kmer{}s\xspace}
\newcommand{\wmer}{$w$-mer\xspace}
\newcommand{\wmers}{\wmer{}s\xspace}
\newcommand{\db}{de~Bruijn\xspace}
\newcommand{\dbg}{B}

\newcommand{\uhs}[2]{\ensuremath{(#1,#2)}-UHS\xspace}
\newcommand{\mks}{\mathcal{M}_{\sigma, w}}
\newcommand{\fbs}{\mathcal{F}_{\sigma, w}}
\newcommand{\re}{\text{Re}}
\newcommand{\im}{\text{Im}}
\newcommand{\wmax}{W_{\text{max}}}
\DeclarePairedDelimiter\norm{\lVert}{\rVert}

\newcommand{\CR}[2]{\ifthenelse{\boolean{camera_ready}}{#1}{#2}}

\newcommand{\cqed}{\ifthenelse{\boolean{camera_ready}}{\qed}{}}
\newcommand{\CRhide}[1]{\CR{}{#1}}

\begin{document}
\title{Lower density selection schemes via small universal hitting sets with short remaining path length}
\titlerunning{Lower density selection schemes}
%
\author{Hongyu~Zheng\inst{1}\orcidID{0000-0002-7668-2090} \and
Carl~Kingsford\inst{1}\orcidID{0000-0002-0118-5516} \and
Guillaume~Mar\c{c}ais\inst{1}\orcidID{0000-0002-5083-5925}}
\authorrunning{H.\ Zheng et al.}
%
\institute{Computational Biology Department, Carnegie Mellon University, Pittsburgh PA 15213, USA\\
  \email{gmarcais@cs.cmu.edu}\\
\url{http://www.cs.cmu.edu/~gmarcais}}

\maketitle              
\begin{abstract}
  Universal hitting sets are sets of words that are unavoidable: every long enough sequence is hit by the set (i.e., it contains a word from the set).
  There is a tight relationship between universal hitting sets and minimizers schemes, where minimizers schemes with low density (i.e., efficient schemes) correspond to universal hitting sets of small size.
  Local schemes are a generalization of minimizers schemes which can be used as replacement for minimizers scheme with the possibility of being much more efficient.
  We establish the link between efficient local schemes and the minimum length of a string that must be hit by a universal hitting set.
  We give bounds for the remaining path length of the Mykkeltveit universal hitting set.
  Additionally, we create a local scheme with the lowest known density that is only a log factor away from the theoretical lower bound.
\end{abstract}

\section{Introduction}

We study the problem of finding \emph{Universal Hitting Sets}~\cite{DOCKS} (UHS).
A UHS is a set of words, each of length $k$, such that every long enough string (say of length $L$ or longer) contains as a substring an element from the set.
We call such a set a universal hitting set for parameters $k$ and $L$.
They are sets of unavoidable words, i.e., words that must be contained in any long strings, and we are interested in the relationship between the size of these sets and the length $L$.

More precisely, we say that a \kmer $a$ (a string of length $k$) \emph{hits} a string $S$ if $a$ appears as a substring of $S$.
A set $A$ of \kmers hits $S$ if at least one \kmer of $A$ hits $S$.
A universal hitting set for length $L$ is a set of \kmers that hits every string of length $L$.
Equivalently, the \emph{remaining path length} of a universal set is the length of the longest string that is not hit by the set ($L-1$ here).


The study of universal hitting sets is motivated in part by the link between UHS and the common method of \emph{minimizers}~\cite{minimizers1,minimizers2,winnowing}.
The minimizers method is a way to sample a string for representative \kmers in a deterministic way by breaking a string into windows, each window containing $w$ \kmers, and selecting in each window a particular \kmer (the ``minimum \kmer'', as defined by a preset order on the \kmers).
This method is used in many bioinformatics software programs (e.g., \cite{kmc2,Mashmap,bcalm2,samsami,sparseassembler}) to reduce the amount of computation and improve run time (see~\cite{sketching-review} for usage examples).
The minimizers method is a family of methods parameterized by the order on the \kmers used to find the minimum.
The \emph{density} is defined as the expected number of sampled \kmers per unit length of sequence.
Depending on the order used, the density varies.

In general, a lower density (i.e., fewer sampled \kmers) leads to greater computational improvements, and is therefore desirable.
For example, a read aligner such a Minimap2~\cite{minimap2} stores all the locations of minimizers in the reference sequence in a database.
It then finds all the minimizers in a read and searches in the database for these minimizers.
The locations of these minimizers are used as seeds for the alignment.
Using a minimizers scheme with a reduced density leads to a smaller database and fewer locations to consider, hence an increased efficiency, while preserving the accuracy.

There is a two-way correspondence between minimizers methods and universal hitting sets: each minimizers method has a corresponding UHS, and a UHS defines a family of \emph{compatible} minimizers methods~\cite{asymptotic-minimizers,improving-minimizers}.
The remaining path length of the UHS is upper-bounded by the number of bases in each window in the minimizers scheme ($L \le w+k-1$).
Moreover, the relative size of the UHS, defined as the size of UHS over the number of possible \kmers, provides an upper-bound on the density of the corresponding minimizers methods: the density is no more than the relative size of the universal hitting set.
Precisely, $\frac{1}{w} \le d \le \frac{|U|}{\sigma^k}$, where $d$ is the density, $U$ is the universal hitting set, $\sigma^k$ is the total number of \kmers on an alphabet of size $\sigma$, and $w$ is the window length.
In other words, the study of universal hitting sets with small size leads to the creation of minimizers methods with provably low density.

Local schemes~\cite{mykkeltveit} and forward schemes are generalizations of minimizers schemes.
These extensions are of interest because they can be used in place of minimizers schemes while sampling \kmers with lower density.
In particular, minimizers schemes cannot have density close to the theoretical lower bound of $1/w$ when $w$ becomes large, while local and forward schemes do not suffer from this limitation~\cite{asymptotic-minimizers}.
Understanding how to design local and forward schemes with low density will allow us to further improve the computation efficiency of many bioinformatics algorithms.

The previously known link between minimizers schemes and UHS relied on the definition of an ordering between \kmers, and therefore is not valid for local and forward scheme that are not based on any ordering.
Nevertheless, UHSs play a central role in understanding the density of local and forward schemes.

Our first contribution is to describe the connection between UHSs, local and forward schemes.
More precisely, there are two connections: first between the density of the schemes and the relative size of the UHS, and second between the window size $w$ of the scheme and the \emph{remaining path length} of the UHS (i.e., the maximum length $L$ of a string that does not contain a word from the UHS).
This motivates our study of the relationship between the size of a universal hitting set $U$ and the remaining path length of $U$.

There is a rich literature on unavoidable word sets (e.g., see~\cite{lothaire2002algebraic}).
The setting for UHS is slightly different for two reasons.
First, we impose that all the words in the set $U$ have the same length $k$, as a \kmer is a natural unit in bioinformatics applications.
Second, the set $U$ must hit any string of a given finite length $L$, rather than being unavoidable only by infinitely long strings.


Mykkeltveit~\cite{mykkeltveit} answered the question of what is the size of a minimum unavoidable set with \kmers by giving an explicit construction for such a set.
The \kmers in the Mykkeltveit set are guaranteed to be present in any infinitely long sequence, and the size of the Mykkeltveit set is minimum in the sense that for any set $\mathcal{S}$ with fewer \kmers there is an infinitely long sequence that avoids $\mathcal{S}$.
On the other hand, the construction gives no indication on the remaining path length.


The DOCKS~\cite{DOCKS} and ReMuVal~\cite{remuval} algorithms are heuristics to generate unavoidable sets for parameters $k$ and $L$.
Both of these algorithms use the Mykkeltveit set as a starting point.
In many practical cases, the longest sequence that does not contain any \kmer from the Mykkeltveit set is much larger than the parameter $L$ of interest (which for a compatible minimizers scheme correspond to the window length).
Therefore, the two heuristics extend the Mykkeltveit set in order to cover every $L$-long sequence.
These greedy heuristics do not provide any guarantee on the size of the unavoidable set generated compared to the theoretical minimum size and are only computationally tractable for limited ranges of $k$ and $L$.

Our second contribution is to give upper and lower bounds on the remaining path length of the Mykkeltveit sets.
These are the first bounds on the remaining path length for minimum size sets of unavoidable \kmers. 



Defining local or forward schemes with density of $O(1/w)$ (that is, within a constant factor of the theoretical lower bound) is not only of practical interest to improve the efficiency of existing algorithms, but it is also interesting for a historical reason.
Both Roberts \emph{et al.}~\cite{minimizers1} and Schleimer \emph{et al.}~\cite{winnowing} used a probabilistic model to suggest that minimizers schemes have an expected density of $2/w$.
Unfortunately, this simple probabilistic model does not correctly model the minimizers schemes outside of a small range of values for parameters $k$ and $w$, and minimizers do not have an $O(1/w)$ density in general.
Although the general question of whether a local scheme with $O(1/w)$ exists is still open, our third contribution is an almost optimal forward scheme with density of $O(\ln(w)/w)$ density.
This is the lowest known density for a forward scheme, beating the previous best density of $O(\sqrt{w}/w)$~\cite{asymptotic-minimizers}, and hinting that $O(1/w)$ might be achievable.

Understanding the properties of universal hitting sets and their many interactions with selection schemes (minimizers, forward and local schemes) is a crucial step toward designing schemes with lower density and improving the many algorithms using these schemes.
In Section~\ref{sec:results}, we give an overview of the results, and in Section~\ref{sec:method} we give \CR{proofs sketches. Full proofs are available in the extended version of the paper on arXiv.}{detailed proofs.}
Further research directions are discussed in Section~\ref{sec:discussion}.

\section{Results}\label{sec:results}

\subsection{Notation}\label{sec:notations}

\paragraph{Universal hitting sets.}

Consider a finite alphabet $\Sigma=\{ 0, \ldots, \sigma-1 \}$ with $\sigma=|\Sigma|$ elements.
If $a\in\Sigma$, $a^k$ denotes the letter $a$ repeated $k$ times.
We use $\Sigma^k$ to denote the set of strings of length $k$ on alphabet $\Sigma$, and call them \kmers.
If $S$ is a string, $S[n,l]$ denotes the substring starting at position $n$ and of length $l$.
For a \kmer $a \in \Sigma^k$ and an $l$-long string $S \in \Sigma^l$, we say ``$a$ hits $S$'' if $a$ appears as substring of $S$ ($a = S[i,k]$ for some $i$).
For a set of \kmers $A\subseteq \Sigma^k$ and $S\in \Sigma^l$, we say ``$A$ hits $S$'' if there exists at least one \kmer in $A$ that hits $S$.
A set $A \subseteq \Sigma^k$ is a universal hitting set for length $L$ if $A$ hits every string of length $L$. 

\paragraph{\db graphs.}
Many questions regarding strings have an equivalent formulation with graph terminology using \emph{\db graphs}.
The \db graph $\dbg_{\Sigma,k}$ on alphabet $\Sigma$ and of order $k$ has a node for every \kmer, and an edge $(u,v)$ for every string of length $k+1$ with prefix $u$ and suffix is $v$.
There are $\sigma^k$ vertices and $\sigma^{k+1}$ edges in the \db graph of order $k$.

There is a one-to-one correspondence between strings and paths in $\dbg_{\Sigma,k}$: a path with $w$ nodes corresponds to a string of $L=w+k-1$ characters.
A universal hitting set $A$ corresponds to a \emph{depathing set} of the \db graph: a universal hitting set for $k$ and $L$ intersects with every path in the \db graph with $w=L-k+1$ vertices.
We say ``$A$ is a $(\alpha,l)$-UHS'' if $A$ is a set of \kmers that is a universal hitting set, with relative size $\alpha = |A|/\sigma^k$ and hits every walk of $l$ vertices (and therefore every string of length $L=l+k-1$).

A \emph{\db sequence} is a particular sequence of length $\sigma^k+k-1$ that contains every possible \kmer once and only once.
Every \db graph is Hamiltonian and the sequence spelled out by a Hamiltonian tour is a \db sequence.

\paragraph{Selection schemes.}

\begin{figure}[h]
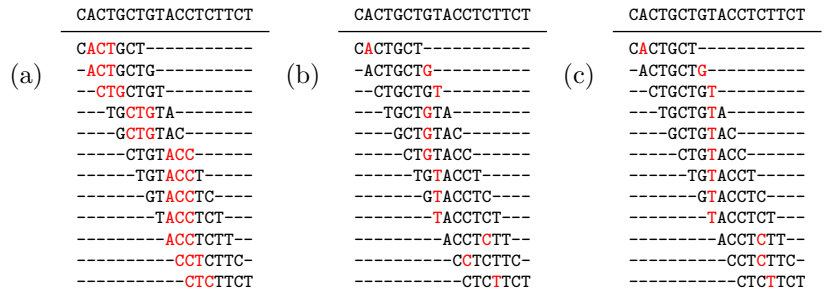

  \newcommand{\M}[1]{\textcolor{red}{#1}}

  \begin{center}
    (a)~\raisebox{-.5\height}{
      \scriptsize
      \begin{tabular}{l}
        \texttt{CACTGCTGTACCTCTTCT}                         \\\midrule
        \texttt{C\M{ACT}GCT-----------}                     \\
        \texttt{-\M{ACT}GCTG----------}                     \\
        \texttt{--\M{CTG}CTGT---------}                     \\
        \texttt{---TG\M{CTG}TA--------}                     \\
        \texttt{----G\M{CTG}TAC-------}                     \\
        \texttt{-----CTGT\M{ACC}------}                     \\
        \texttt{------TGT\M{ACC}T-----}                     \\
        \texttt{-------GT\M{ACC}TC----}                     \\
        \texttt{--------T\M{ACC}TCT---}                     \\
        \texttt{---------\M{ACC}TCTT--}                     \\
        \texttt{----------\M{CCT}CTTC-}                     \\
        \texttt{-----------\M{CTC}TTCT}                     \\\midrule
      \end{tabular}
    } (b)~\raisebox{-.5\height}{
      \scriptsize
      \begin{tabular}{l}
        \texttt{CACTGCTGTACCTCTTCT}                         \\\midrule
        \texttt{C\M{A}CTGCT-----------}                     \\
        \texttt{-ACTGCT\M{G}----------}                     \\
        \texttt{--CTGCTG\M{T}---------}                     \\
        \texttt{---TGCT\M{G}TA--------}                     \\
        \texttt{----GCT\M{G}TAC-------}                     \\
        \texttt{-----CT\M{G}TACC------}                     \\
        \texttt{------TG\M{T}ACCT-----}                     \\
        \texttt{-------G\M{T}ACCTC----}                     \\
        \texttt{--------\M{T}ACCTCT---}                     \\
        \texttt{---------ACCT\M{C}TT--}                     \\
        \texttt{----------C\M{C}TCTTC-}                     \\
        \texttt{-----------CTC\M{T}TCT}                     \\\midrule
      \end{tabular}
    } (c)~\raisebox{-.5\height}{ 
      \scriptsize
      \begin{tabular}{l}
        \texttt{CACTGCTGTACCTCTTCT}                         \\\midrule
        \texttt{C\M{A}CTGCT-----------}                     \\
        \texttt{-ACTGCT\M{G}----------}                     \\
        \texttt{--CTGCTG\M{T}---------}                     \\
        \texttt{---TGCTG\M{T}A--------}                     \\
        \texttt{----GCTG\M{T}AC-------}                     \\
        \texttt{-----CTG\M{T}ACC------}                     \\
        \texttt{------TG\M{T}ACCT-----}                     \\
        \texttt{-------G\M{T}ACCTC----}                     \\
        \texttt{--------\M{T}ACCTCT---}                     \\
        \texttt{---------ACCT\M{C}TT--}                     \\
        \texttt{----------CCT\M{C}TTC-}                     \\
        \texttt{-----------CTC\M{T}TCT}                     \\\midrule
      \end{tabular}
    }
  \end{center}

  \caption{(a)~Example of selecting minimizers with $k=3$, $w=5$ and the lexicographic order (i.e., $\texttt{AAA} < \texttt{AAC} < \texttt{AAG} < \ldots < \texttt{TTT}$).
    The top line is the input sequence, each subsequent line is a $7$-bases long window (the number of bases in a window is $w+k-1=7$) with the minimum $3$-mer highlighted.
    The positions $\{1, 2, 5, 9, 10, 11\}$ are selected for a density $d=6/(18-3+1)=0.375$.
    (b)~On the same sequence, an example of a selection scheme for $w=7$ (and $k=1$ because it is a selection scheme, hence the number of bases in a window is also $w$).
    The set of positions selected is $\{ 1, 6, 7, 8, 11, 13, 14 \}$.
    This is not a forward scheme as the sequence of selected position is not non-decreasing.
    (c)~A forward selection scheme for $w=7$ with selected positions $\{ 1, 7, 8, 12, 13\}$.
    Like the minimizers scheme, the sequence of selected positions is non-decreasing.
  }
  \label{fig:minimizer-ex}
\end{figure}

A \emph{local scheme}~\cite{winnowing} is a method to select positions in a string.
A local scheme is parameterized by a \emph{selection function} $f$.
It works by looking at every \wmer of the input sequence $S$:  $S[0,w], S[1,w], \ldots$, and selecting in each window a position according to the selection function $f$.
The selection function selects a position in a window of length $w$, i.e., it is a function $f: \Sigma^{w} \rightarrow [0:w-1]$.
The output of a forward scheme is a set of selected positions: $\{ i + f(S[i,w]) \mid 0 \le i < |S|-w \}$.

A \emph{forward scheme} is a local scheme with a selection function such that the selected positions form a non-decreasing sequence.
That is, if $\omega_1$ and $\omega_2$ are two consecutive windows in a sequence $S$, then $f(\omega_2) \ge f(\omega_1)-1$.

A \emph{minimizers scheme} is scheme where the selection function takes in the sequence of $w$ consecutive \kmers and returns the ``minimum'' \kmer in the window (hence the name minimizers).
The minimum is defined by a predefined order on the \kmers (e.g., lexicographic order) and the selection function is $f: \Sigma^{w+k-1}\rightarrow [0:w-1]$.

See Figure~\ref{fig:minimizer-ex} for examples of all 3 schemes.
The local scheme concept is the most general as it imposes no constraint on the selection function, while a forward scheme must select positions in a non-decreasing way.
A minimizers scheme is the least general and also selects positions in a non-decreasing way.

Local and forward schemes were originally defined with a function defined on a window of $w$ \kmers, $f: \Sigma^{w+k-1} \rightarrow [0:w-1]$, similarly to minimizers.
Selection schemes are schemes with $k=1$, and have a single parameter $w$ as the word length.
While the notion of \kmer is central to the definition of the minimizers schemes, it has no particular meaning for a local or forward scheme: these schemes select positions within each window of a string $S$, and the sequence of the \kmers at these positions is no more relevant than sequence elsewhere in the window to the selection function.

There are multiple reasons to consider selection schemes.
First, they are slightly simpler as they have only one parameter, namely the window length $w$.
Second, in our analysis we consider the case where $w$ is asymptotically large, therefore $w \gg k$ and the setting is similar to having $k=1$.
Finally, this simplified problem still provides information about the general problem of local schemes.
Suppose that $f$ is the selection function of a selection scheme, for any $k>1$ we can define $g_k:\Sigma^{w+k-1}\rightarrow [0,w-1]$ as $g_k(\omega)=f(\omega[0,w])$.
That is, $g_k$ is defined from the function $f$ by ignoring the last $k-1$ characters in a window.
The functions $g_k$ define proper selection functions for local schemes with parameter $w$ and $k$, and because exactly the same positions are selected, the density of $g_k$ is equal to the density of $f$.
In the following sections, unless noted otherwise, we use forward and local schemes to denote forward and local selection schemes. 

\paragraph{Density.}
Because a local scheme on string $S$ may pick the same location in two different windows, the number of selected positions is usually less than $|S|-w+1$.
The \emph{particular density} of a scheme is defined as the number of distinct selected positions divided by $|S|-w+1$ (see Figure~\ref{fig:minimizer-ex}).
The \emph{expected density}, or simply the \emph{density}, of a scheme is the expected density on an infinitely long random sequence.
Alternatively, the expected density is computed exactly by computing the particular density on any \db sequence of order $\ge 2w-1$.
In other words, a \db sequence of large enough order ``looks like'' a random infinite sequence with respect to a local scheme (see \cite{improving-minimizers} and Section~\ref{sec:uhs-from-selection}).

\subsection{Main Results}\label{sec:result_overview}

The density of a local scheme is in the range $[1/w,1]$, as $1/w$ corresponds to selecting exactly one position per window, and $1$ corresponds to selecting every position.
Therefore, the density goes from a low value with a constant number of positions per window (density is $O(1/w)$, which goes to $0$ when $w$ gets large), to a high with constant value (density is $\Omega(1)$) where the number of positions per window is proportional to $w$.
When the minimizers and winnowing schemes were introduced, both papers used a simple probabilistic model to estimate the expected density to $2/(w+1)$, or about $2$ positions per window.
Under this model, this estimate is within a constant factor of the optimal, it is $O(1/w)$.

Unfortunately, this simple model properly accounts for the minimizers behavior only when $k$ and $w$ are small.
For large $k$ ---i.e., $k \gg w$--- it is possible to create almost optimal minimizers scheme with density $\sim 1/w$.
More problematically, for large $w$ ---i.e., $w \gg k$--- and for all minimizer schemes the density becomes constant ($\Omega(1)$)~\cite{asymptotic-minimizers}.
In other words, minimizers schemes cannot be optimal or within a constant factor of optimal for large $w$, and the estimate of $2/(w+1)$ is very inaccurate in this regime.

This motivates the study of forward schemes and local schemes.
It is known that there exists forward schemes with density of $O(1/\sqrt{w})$~\cite{asymptotic-minimizers}.
This density is not within a constant factor of the optimal density but at least shows that forward and local schemes do not have constant density like minimizers schemes for large $w$ and that they can have much lower density.

\paragraph{Connection between UHS and selection schemes.}

In the study of selection schemes, as for minimizers schemes, universal hitting sets play a central role.
We describe the link between selection schemes and UHS, and show that the existence of a selection scheme with low density implies the existence of a UHS with small relative size.

\begin{restatable}{theorem}{uhslocal} \label{tm:uhslocal}
  Given a local scheme $f$ on \wmers with density $d_f$, we can construct a $\uhs{d_f}{w}$ on $(2w-1)$-mers.
  If $f$ is a forward scheme, we can construct a $\uhs{d_f}{w}$ on $(w+1)$-mers.
\end{restatable}


\paragraph{Almost-optimal relative size UHS for linear path length.}

Conversely, because of their link to forward and local selection schemes, we are interested in universal hitting set with remaining path length $O(w)$.
Necessarily a universal hitting hits any infinitely long sequences.
On \db graphs, a set hitting every infinitely long sequences is a \emph{decycling set}: a set that intersects with every cycle in the graph.
In particular, a decycling set must contain an element in each of the cycles obtained by the rotation of the \wmers (e.g., cycle of the type $001 \rightarrow 010 \rightarrow 100 \rightarrow 001$).
The number of these rotation cycles is known as the ``necklace number'' $N_{\sigma, w} = \frac{1}{n}\sum_{d|w} \varphi(d)\sigma^{w/d} = O(\sigma^w/w)$~\cite{golomb}, where $\varphi(d)$ is the Euler's totient function.

Consequently, the relative size of a UHS, which contains at least one element from each of these cycles, is lower-bounded by $O(1/w)$.
The smallest previously known UHS with $O(w)$ remaining path length has a relative size of $O(\sqrt{w}/w)$~\cite{asymptotic-minimizers}.
We construct a smaller universal hitting set with relative size $O(\ln(w)/w)$:
\begin{restatable}{theorem}{uhslinear}
  For every sufficiently large $w$, there is a forward scheme with density of $O(\ln(w)/w)$ and a corresponding \uhs{O(\ln(w)/w)}{w}.
\end{restatable}


\paragraph{Remaining path length bounds for the Mykkeltveit sets.}

Mykkeltveit~\cite{mykkeltveit} gave an explicit construction for a decycling set with exactly one element from each of the rotation cycles, and thereby proved a long standing conjecture~\cite{golomb} that the minimal size of decycling sets is equal to the necklace number.
Under the UHS framework, it is natural to ask what the remaining path length for Mykkeltveit sets is.
Given that the \db graph is Hamiltonian, there exists paths of length exponential in $w$: the Hamiltonian tours have $\sigma^w$ vertices.
Nevertheless, we show that the remaining path length for Mykkeltveit sets is upper- and lower-bounded by polynomials of $w$:

\begin{restatable}{theorem}{mkbounds}
For sufficiently large $w$, the Mykkeltveit set is a \uhs{N_{\sigma, w}/\sigma^w}{g(w)}, having the same size as minimal decycling sets, while $g(w)=O(w^3)$ and $g(w) > cw^2$ for some constant $c$.
\end{restatable}

\section{Methods and Proofs}\label{sec:method}
\CR{
Due to page limits, we provide proof sketches of the results.
Full proofs are available in the extended paper on arXiv: [link].
}{
For simplicity, several parts of the proof are found in the Supplementary materials.}

\subsection{UHS from Selection Schemes} \label{sec:uhs-from-selection}

\subsubsection{Contexts and densities of selection schemes}
In this section, we derive another way of calculating densities of selection schemes based on the idea of \emph{contexts}.

Recall a local scheme is defined as a function $f: \Sigma^w \rightarrow [0, w-1]$.
For any sequence $S$ and scheme $f$, the set of selected locations are $\{f(S[i, w])+i\}$ and the density of $f$ on the sequence is the number of selected locations divided by $|S|-w+1$.
Counting the number of distinct selected locations is the same as counting the number of \wmers $S[i, w]$ such that $f$ picks a new location from all previous \wmers.
$f$ can pick identical locations on two \wmers only if they overlap, so intuitively, we only need to look back $(w-1)$ windows to check if the position is already picked.
Formally, $f$ picks a new position in window $S[i, w]$ if and only if $f(S[i, w]) + i \neq f(S[i-d, w]) + (i-d)$ for all $1\leq d \leq w-1$.

For a location $i$ in sequence $S$, the context at this location is defined as $c_i = S[i-w+1, 2w-1]$, a $(2w-1)$-mer whose last \wmer starts at $i$.
Whether $f$ picks a new position in $S[i, w]$ is entirely determined by its context, as the conditions only involve \wmers as far back as $S[i-w+1, w]$, which are included in the context.
This means that instead of counting selected positions in $S$, we can count the contexts $c$ satisfying $f(c[w-1, w]) + w-1 \neq f(c[j, w]) + j$ for all $0\leq j \leq w-2$, which are the contexts such that $f$ on the last \wmer of $c$ picks a new location.
We define $\mathcal{C}_f \subset \Sigma^{2w-1}$ the set of contexts that satisfy this condition, formally:

\begin{definition} For given $w$ and local selection scheme $f:\Sigma^w\rightarrow [0, w-1]$, $\mathcal{C}_f=\{c\in \Sigma^{2w-1} \mid \forall 0\leq i\leq w-2, f(c[w-1,w]) + (w-1) \neq f(c[i, w]) + i \}$ is a subset of $\Sigma^{2w-1}$.
\end{definition}

The expected density of $f$ is computed as the number of selected positions over the length of the sequence for a random sequence, as the sequence becomes infinitely long.
For a sufficiently long random sequence ($|S| \gg w$), the distribution of its contexts converges to a uniform random distribution over $(2w-1)$-mers.
Because the distribution of these contexts is exactly equal to the uniform distribution on a circular \db $S$ sequence of order at least $2w-1$, we can calculate the expected density of $f$ as the density of $f$ on $S$, or as $|\mathcal{C}_f|/\sigma^{2w-1}$.

\subsubsection{UHS from local selection schemes}

We now prove that $\mathcal{C}_f$ over $(2w-1)$-mers is the UHS we need for Theorem~\ref{tm:uhslocal}.
\begin{lemma} \label{lm:uhslocal}
$\mathcal{C}_f$ is a UHS with remaining path length of at most $w-1$.
\end{lemma}

\begin{proof}
By contradiction, assume there is a path of length $w$ in the \db graph of order $(2w-1)$, say $\{c_0, c_1, \cdots, c_{w-1}\}$, that avoids $\mathcal{C}$.
We construct the sequence $S'$ corresponding to the path: $S'\in \Sigma^{3w-2}$ such that $S'[i, 2w-1] = c_i$.

Since $c_{w-1}\notin \mathcal{C}$ and $S'$ includes $c_{w-1}$, it means $f$ on the last \wmer of $c_{w-1}$ (which is $S'[2w-2,w]$) picks a location that has been picked before on $S'$.
The coordinate $l$ of this selection in $S'$ satisfies $l \geq 2w-2$.
As $0 \leq f(x) \leq w-1$, the first \wmer $S'[m,w]$ in $S'$ such that $f$ picks $S'[l]$ (that is, $m+f(S'[m,w]) = l$) satisfies $m\geq w-1$.
The context $c_{m-w+1} = S'[m-(w-1), 2w-1]$ then satisfies that a new location $l$ is picked when $f$ is applied to its last \wmer, and by definition $c_{m-w+1} \in \mathcal{C}$, contradiction.
\cqed \end{proof}

This results is also a direct consequence of the definition of $\mathcal{C}$.
Details can be found in \CR{full version of this paper.}{Supplementary Section~\ref{supp:uhslocal}.}

\subsubsection{UHS from forward selection schemes}

When $f$ is a forward scheme, to determine if a new location is picked in a window, looking back one window is sufficient.
This is because if we do not pick a new location, we have to pick the same location as in last window.
This means context with two \wmers, or as a $(w+1)$-mer, is sufficient, and our other arguments involving contexts still hold.
Combining the pieces, we prove the following theorem:
\uhslocal*

\subsection{Forbidden Word Depathing Set} \label{ssc:forbidden_word_set}
\subsubsection{Construction and path length}
In this section, we prove the following set is a $\uhs{O(\ln(w)/w)}{w}$.
\begin{restatable}[Forbidden Word UHS]{definition}{dfforbiddenword} \label{df:forbidden_set}
  Let $d=\lfloor \log_\sigma (w/\ln(w)) \rfloor-1$.
  Define $\fbs$ as the set of \wmers that satisfies either of the following clauses: (1)~$0^d$ is the prefix of $x$ (2)~$0^d$ is not a substring of $x$.
\end{restatable}

Note that in this theorem we treat $\sigma$ as a constant and assume $\sigma \ge 2$.
We also assume that $w$ is sufficiently large such that $d\geq 1$.

\begin{lemma}
  The longest remaining path in the \db graph of order $w$ after removing $\fbs$ is $w-d$.
\end{lemma}
\CRhide{\begin{proof}
Let $\{x_0, x_1,\cdots, x_{w-d}\}$ be a path of length $w-d+1$ in the \db graph.
If $x_0$ does not have a substring equal to $0^d$, it is in $\fbs$.
Otherwise, let $c$ be the index such that $x_0[c,d]=0^d$.
Since $c\leq w-d$, $x_c[0,d]=0^d$ and $x_c$ is in $\fbs$.

On the other hand, let $S=1^{w-d}0^{d}1^{w-d-1} \in \Sigma^{2w-d-1}$ and $x_i = S[i, w]$ for $0\leq i < w-d$.
None of $\{x_i\}$ is in $\fbs$, meaning there is a path of length $w-d$ in the remaining graph.
\cqed \end{proof}}

\begin{lemma}
  The relative size of the set of \wmers satisfying clause~1 of Definition~\ref{df:forbidden_set} is $O(\ln(w)/w)$.
\end{lemma}
\begin{proof}
  The number of \wmer satisfying clause~1 is $\sigma^{w-d} = O(\ln(w)\sigma^w/w)$.
\cqed \end{proof}

For the rest of this section, we focus on counting \wmers satisfying clause 2 in Definition~\ref{df:forbidden_set}, that is, the number of \wmers not containing $0^d$.
We employ a finite state machine based approach.

\subsubsection{Number of \wmers not containing $0^d$}
We construct a finite state machine (FSM) that recognizes $0^{d}$ as follows.
The FSM consists of $d+1$ states labeled ``$0$'' to ``$d$'', where ``$0$'' is the initial state and ``$d$'' is the terminal state.
The state ``$i$'' with $0 \le i \le d-1$ means that the last $i$ characters were $0$ and $d-i$ more zeroes are expected to match $0^d$.
The terminal state ``$d$'' means that we have seen a substring of $d$ consecutive zeroes.
If the machine is at non-terminal state ``$i$'' and receives the character $0$, it moves to state ``$i+1$'', otherwise it moves to state ``$0$''; once the machine reaches state ``$d$'', it remains in that state forever.

Now, assume we feed a random \wmer to the finite state machine.
The probability that the machine does not reach state ``$d$'' for the input \wmer is the relative size of the set of \wmer satisfying clause 2.
Denote $p_k\in\mathbb{R}^d$ such that $p_k(j)$ is the probability of feeding a random $k$-mer to the machine and ending up in state ``$j$'', for $0\leq j<d$ (note that the vector does not contain the probability for the terminal state ``$d$'').
The answer to our problem is then $\norm{p_w}_1 = \sum_{i=0}^{d-1} p_w(i)$, that is, the sum of the probabilities of ending at a non-terminal state. 

Define $\mu = 1/\sigma$.
Given that a randomly chosen \wmer is fed into the FSM, i.e., each base is chosen independently and uniformly from $\{0, 1, \cdots, \sigma-1\}$, the probabilities of transition in the FSM are: ``$i$'' $\rightarrow$ ``$i+1$'' with probability $1/\sigma=\mu$, ``$i$'' $\rightarrow$ ``$0$'' with probability $1-\mu$.
The (partial) probability matrix to recognize $0^d$ is a $d \times d$ matrix, as we discard the row and column associated with terminal state ``$d$'':
\begin{equation*}
  A_d =
  \begin{bmatrix}
    1-\mu  & 1-\mu  & \dots  & 1-\mu  & 1-\mu  \\
    \mu    & 0      & \dots  & 0      & 0      \\
    0      & \mu    & \dots  & 0      & 0      \\
    \vdots & \vdots & \ddots & \vdots & \vdots \\
    0      & 0      & \dots  & \mu    & 0
  \end{bmatrix}_{d\times d} =
  \begin{bmatrix}
    (1-\mu)\mathbf{1}^T_{d-1} & 1-\mu \\ 
    \mu \mathbf{\mathrm{I}}_{d-1}                & \mathbf{0}_{d-1} 
  \end{bmatrix}
\end{equation*}

Starting with $p_0=(1, 0, \ldots, 0) \in \mathbb{R}^d$ as initially no sequence has been parsed and the machine is at state ``$0$'' with probability 1, we can compute the probability vector $p_w$ as $p_w=A_dp_{w-1}=A_d^wp_0$.

\subsubsection{Bounding $\norm{p_w}_1$} \label{ssc:proof_fbseigenvalues}
We start by deriving the characteristic polynomial $p_{A_d}(\lambda)$ of $A_d$ and its set of roots (which are the eigenvalues of $A_d$):
\begin{restatable}{lemma}{lmcharpolynomial}\label{lem:char-polynomial}
  \begin{equation*}
    p_{A_d}(\lambda)=\det(A_d-\lambda I) = 
    \begin{cases}
        (-1)^d\frac{\lambda^{d+1}-\lambda^d-\mu^{d+1}+\mu^d}{\lambda-\mu} & \lambda \neq \mu \\ 
        (-\mu)^{d-1}((1-\mu)d-\mu)                                        & \lambda = \mu
    \end{cases}
  \end{equation*}
\end{restatable}

\CRhide{\begin{proof}
The characteristic polynomial of $A_d$ satisfies the following recursive formula, obtained by expanding the determinant over the first column and using the linearity of the determinant:
\begin{align*}
  \det(A_d-\lambda I_d) =&\ 
                           \begin{vmatrix}
                             1-\mu-\lambda & 1-\mu & 1-\mu & \cdots & 1-\mu \\
                             \mu & -\lambda & 0 & \cdots & 0 \\
                             0 & \mu & -\lambda & \cdots & 0 \\
                             \vdots & \vdots & \vdots & \ddots & \vdots \\
                             0 & 0 & 0 &\cdots & -\lambda
                           \end{vmatrix}_{d\times d} \\
  =&  (1-\lambda) (-\lambda)^{d-1} -\mu p_{A_{d-1}}(\lambda).
\end{align*}


For $d=1$, we have $p_{A_1}(\lambda) = 1 - \mu - \lambda$.
Assuming $\lambda\neq \mu$ for now, we repeatedly expand the recursive formula to obtain a closed form formula for $p_{A_d}(\lambda)$:
\begin{align*}
p_{A_d}(\lambda) & = (1-\lambda) \left[ (-\lambda)^{d-1}+(-\mu)^1(-\lambda)^{d-2}+\cdots + (-\mu)^{d-2}(-\lambda)^1+(-\mu)^{d-1} \right] + (-\mu)^{d} & (*) \\
                 & = (-1)^d \left[ (\lambda-1)\frac{\lambda^d-\mu^d}{\lambda-\mu} + \mu^d \right]                                                                              \\
                 & = (-1)^d \frac{\lambda^{d+1}-\lambda^d-\mu^{d+1}+\mu^d}{\lambda-\mu}
\end{align*}
The value for the characteristic polynomial when $\lambda=\mu$ can be derived by plugging $\lambda=\mu$ in the line marked with $(*)$ to obtain $p_{A_d}(\lambda) = (1 - \mu)d\mu^{d-1}+(-\mu)^d$.
 \cqed \end{proof}}

Now we fix $d$ and focus on the polynomial $f_d(\lambda)=\lambda^{d+1}-\lambda^d-\mu^{d+1}+\mu^d$. 
Since this is a polynomial of degree $d+1$, it has $d+1$ roots and except for $\mu$, which is a root of $f_d$ but not of $p_{A_d}$, $f_d$ and $p_{A_d}$ have the same roots.

\begin{lemma} \label{slm:largeeigen}
For sufficiently large $d$, $f_d(\lambda)$ has a real root $\lambda_0$ satisfying $1-\mu^{d} < \lambda_0 < 1-\mu^{d+1}$.
\end{lemma}
\CRhide{\begin{proof}
  We show $f_d$ has opposite signs on the lower and upper bound of this inequality for sufficiently large $d$.
  \begin{align*}
    f_d(1-\mu^{d})   & = (1-\mu^{d})^{d+1}-(1-\mu^{d})^{d}-\mu^{d+1}+\mu^d                   \\
                     & = 1-(d+1)\mu^{d}+O(\mu^{2d})-1+d\mu^d-O(\mu^{2d})-\mu^{d+1}+\mu^d     \\
                     & = -\mu^{d+1} + O(\mu^{2d}) < 0                                                                   \\
    f_d(1-\mu^{d+1}) & =(1-\mu^{d+1})^{d+1}-(1-\mu^{d+1})^{d}-\mu^{d+1}+\mu^d                \\
                     & = 1-(d+1)\mu^{d+1}+C(d+1, 2)\mu^{2d+2}+O(\mu^{3d+6})                  \\
                     & \ \ \ \ -1+d\mu^{d+1}-C(d, 2)\mu^{2d+2}-O(\mu^{3d+6})-\mu^{d+1}+\mu^d \\
                     & = -2\mu^{d+1}+\mu^{d}+d\mu^{2d+2}+O(\mu^{3d+6})  > 0
  \end{align*}
  For the last line, if $\sigma=2$ the first two terms cancel out and $d\mu^{2d+2}$ becomes dominant and positive, otherwise $\mu^{d}=\sigma \mu^{d+1}>2\mu^{d+1}$.
  Since $f_d$ is polynomial, $f_d$ is continuous and thus has a root between $1-\mu^{d}$ and $1-\mu^{d+1}$.
\cqed \end{proof}}

\begin{lemma}\label{lem:norm-eigenvector}
Let $s=\mu/\lambda_0$.
$\nu_0$ = $(1, s, s^2, \cdots, s^{d-1})$ is the right eigenvector of $A_d$ corresponding to eigenvalue $\lambda_0$, and $\norm{\nu_0}_1 < 3$ for sufficiently large $d$.
\end{lemma}
\CRhide{\begin{proof}
  For the first part, we need to verify $A_d\nu_0 = \lambda_0\nu_0$.
  For indices $1\leq i < d$, $(A_d\nu_0)_i = \mu(\nu_0)_{i-1} = \mu s^{i-1} = \lambda_0 s^i = (\lambda_0 \nu_0)_i$.
  For the first element in the vector, we have:
  \begin{align*}
    (A_d\nu_0)_0 - (\lambda_0\nu_0)_0 &= (1-\mu)(1+s+\cdots s^{d-1}) - \lambda_0 \\
                                      &= \frac{(1-\mu)(s^d-1)-\lambda_0(s-1)}{s-1} \\
                                      &= \frac{\lambda_0^d(\mu^d - \lambda_0^d - \mu^{d+1} + \lambda_0^{d+1})}{s-1} \\
                                      &= \frac{\lambda_0^d  f_d(\lambda_0)}{s-1} = 0.
  \end{align*}
  This verifies $A_d\nu_0 = \lambda_0\nu_0$.
  For the second part, note that for sufficiently large $d$ we have $\lambda_0 > 1- \mu^d > 0.9$ and since $\mu \leq 0.5$, we have $s = \mu/\lambda_0 < 2/3$.
  Every element of $\nu_0$ is positive, so $\norm{\nu_0}_1 = \sum_{i=0}^{d-1} s^i < \sum_{i=0}^{\infty} s^i = 1/(1-s) < 3$.
\cqed \end{proof}}

\begin{restatable}{lemma}{lmnormbounds}
  $\norm{p_w}_1 = \norm{A_d^wp_0}_1 = O(1/w)$.
\end{restatable}

\begin{proof}

  Let $\eta_0 = \nu_0 - p_0 = (0, s, s^2, \cdots, s^{d-1})$, where $s=\mu/\lambda_0$ from last lemma.
  Because $\lambda_0 > 0$, the elements of $\eta_0$ and $A_d$ are all nonnegative, then the elements of $A_d^w\eta_0$ and $\lambda_0 \eta_0$ are also nonnegative.
  Now, recall that $d=\lfloor \log_\sigma (w/\ln(w)) \rfloor-1$, which implies that $\mu^{d+1} \geq \ln(w)/w$.
  \begin{align*}
    \norm{p_w}_1 & = \norm{A_d^wp_0}_1                                                                                 \\
                 & = \norm{A_d^w(\nu_0 - \eta_0)}_1                                                                    \\
                 & = \norm{\lambda_0^w \nu_0 - A_d^w \eta_0}_1     & \text{($\nu_0$ is a eigenvector of  $A_d$)}                                                    \\
                 & \leq \lambda_0^w \norm{\nu_0}_1 & \text{(nonnegative elements)}                                        \\
                 & < 3(1-\mu^{d+1})^w           & \text{(by Lemma~\ref{slm:largeeigen} \& \ref{lem:norm-eigenvector})} \\
                 & \leq 3(1-\ln(w)/w)^w         & \text{(by definition of $d$)}                                        \\
                 & \leq 3\exp(-\ln(w))                & \text{($1-x\leq e^{-x}$)}                                                         \\
                 & = O(1/w).
  \end{align*}
\cqed \end{proof}

These lemma implies that the relative size for the set $\fbs$ is dominated by the \wmers satisfying clause~1 of Definition~\ref{df:forbidden_set} and $\fbs$ is of relative size $O(\ln(w)/w)$.
This completes the proof that $\fbs$ is a $\uhs{O(\ln(w)/w)}{w}$.

\subsection{Construction of the Mykkeltveit sets}

\begin{figure}
  \centering
  \raisebox{-.5\height}{\includegraphics{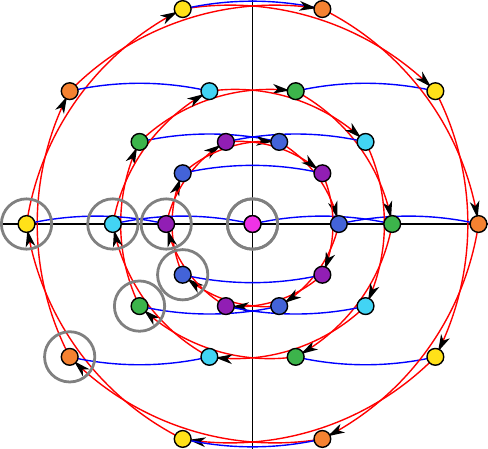}} \hfill \raisebox{-.5\height}{\includegraphics{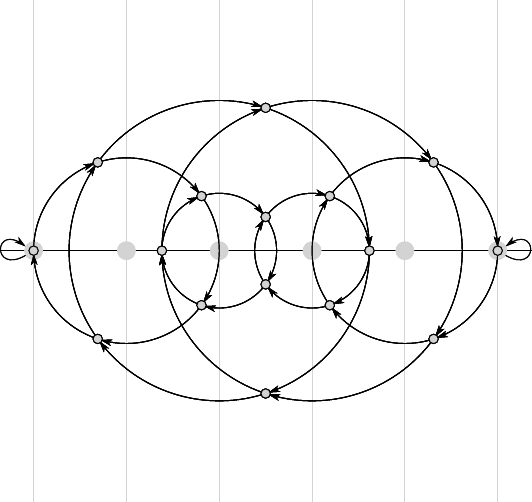}}
  \caption{(a)~Mykkeltveit embedding of the \db graph of order $5$ on the binary alphabet.
    The nodes of a conjugacy class have the same color and form a circle (there is more than one class per circle).
    The pure rotations are represented by the red edges.
    A non-pure rotation $S_a(x)$ is a red edge followed by a horizontal shift (blue edge).
    The set of nodes circled in gray is the Mykkeltveit  set.
    (b)~Weight-in embedding of the same graph.
    Multiple \wmers map to the same position in this embedding and each circle represent a conjugacy class.
    The gray dots on the horizontal axis are the $w$ centers of rotations and the vertical gray lines going through the centers separate the space in sub-regions of interest.
  }
  \label{fig:embedding}
\end{figure}

In this section, we construct the Mykkeltveit set $\mks$ and prove some important properties of the set.
We start with the definition of the Mykkeltveit embedding of the \db graph.
\begin{definition}[Modified Mykkeltveit Embedding] \label{df:modified_embeddiing}
For a \wmer $x$, its embedding in the complex plane is defined as $P(x) = \sum_{i=0}^{w-1} x_i r_w^{i+1}$, where $r_w$ is a $w^{\text{th}}$ root of unity, $r_w = e^{2\pi i/w}$.
\end{definition}

Intuitively, the position of a \wmer $x$ is defined as the following center of mass.
The $w$ roots of unity form a circle in the complex plane, and a weight equal to the value of the base $x_i$ is set at the root $r_w^{i+1}$.
The position of $x$ is the center of mass of these $w$ points and associated weights.
Originally, Mykkeltveit defined the embedding with weight $r_w^i$~\cite{mykkeltveit}.
This extra factor of $r_w$ in our modified embedding rotates the coordinate and is instrumental in the proof.

We now focus on a particular kind of cycle in the \db graph.
The \textbf{pure cycles} in the \db graph, also known as \textbf{conjugacy classes}, are the set of cycles $x_0x_1\cdots x_{w-1} \rightarrow x_1x_2\cdots x_{w-1}x_0 \rightarrow x_2x_3 \cdots x_0x_1 \rightarrow \cdots \rightarrow x_{w-1}x_0\cdots x_{w-3}x_{w-2} \rightarrow x_0x_1\cdots x_{w-1}$.
We use $[x]$ to denote the set of \wmers in the pure cycle containing $x$.
Each cycle consists of $w$ \wmers, unless $x_0x_1\cdots x_{w-1}$ is periodic, and in this case the size of cycle is equal to its shortest period.

Define the successor function $S_a(x) = x_1x_2\cdots x_{w-1} a$.
The successor function gives all the neighbors of $x$ in the \db graph.
For any \wmer $x$, we have $[S_a(x)]=[x]$ if and only if $a=x_0$.
We call this kind of moves a \textbf{pure rotation}, and use $R(x)$ to denote the resulting \wmer.
The embeddings from pure rotations satisfy a curious property:

\begin{lemma}[Rotations and Embeddings] \label{lm:myk_successor}
$P(R(x))$ on the complex plane is $P(x)$ rotated clockwise around origin by $2\pi/w$.
$P(S_a(x))$ is $P(R(x))$ shifted by $\delta = a-x_0$ on the real part, with the imaginary part unchanged.
\end{lemma}

\begin{proof}
By Definition~\ref{df:modified_embeddiing} and the the definition of successor function $S_a(x)$:
\begin{align*}
P(S_a(x)) &= {\textstyle\sum}_{i=0}^{w-1} (S_a(x))_i r_w^{i+1} \\
&= {\textstyle\sum}_{i=0}^{w-2} x_{i+1}r_w^{i+1} + ar_w^{w-1+1} \\
&= r_w^{-1}{\textstyle\sum}_{i=0}^{w-1} x_{i} r_w^{i+1} + (a-x_0) \\
&= r_w^{-1}P(x)+\delta
\end{align*}
Note that for pure rotations $\delta=0$, and $r_w^{-1}P(x)$ is exactly $P(x)$ rotated clockwise by $2\pi/w$.
 \cqed \end{proof}

 The range for $\delta$ is $[-\sigma+1, \sigma-1]$.
 In particular, $\delta$ can be negative.
 This implies that in a pure cycle either all \wmer satisfy $P(x)=0$, or they lie equidistant on a circle centered at origin.
Figure~\ref{fig:embedding}(a) shows the embeddings and pure cycles of 5-mers.
It is known that we can partition the set of all \wmers into $N_{\sigma, k}$ pure cycles, and these cycles are disjoints.
This means any decycling set that breaks every cycle of the \db graph will be at least this large.
We now construct our proposed depathing set with this idea in mind.

\begin{definition}[Mykkeltveit Set] \label{df:depathing1} 
We construct the Mykkeltveit set $\mks$ as follows.
Consider each conjugacy class $[x]$, we will pick one \wmer from each of them by the following rule:
\begin{enumerate}
\item If every \wmer in the class embeds to the origin, pick an arbitrary one.
\item If there is one \wmer $x$ in the class such that $\re(P(x))<0$ and $\im(P(x))=0$, pick that one.
\item Otherwise, pick the unique \wmer $x$ such that $\im(P(x))<0$ and $\im(P(R(x)))>0$. Intuitively, this is the \wmer in the cycle right below the negative real axis.
\end{enumerate}
\end{definition}

This set breaks every pure cycle in the \db graph by its construction, with an interesting property as follows:
\begin{restatable}{lemma}{lmmykoneway} \label{lm:myk_oneway}
Let $\{x_i\}$ be a path on the \db graph that avoids $\mks$.
If $\im(P(x_i)) \leq 0$, then for all $j\geq i$, $\im(P(x_j)) \leq 0$.
\end{restatable}

\CRhide{\begin{proof} It suffices to show that in the remaining \db graph after removing $\mks$, there are no edges $x\rightarrow y$ such that $\im(P(x)) \leq 0$ and $\im(P(y)) > 0$.
  The edge $x\rightarrow y$ means that $y = S_a(x)$ for some $a$.
  By Lemma~\ref{lm:myk_successor}, $\im(P(R(x))) = \im(P(S_a(x))) = \im(P(y)) > 0$.
\begin{itemize}
\item If we have $\im(P(x)) < 0$, by clause 3 of Definition~\ref{df:depathing1}, $x\in \mks$.
\item If we have $\im(P(x)) = 0$ and $\re(P(x)) < 0$, by clause 2 of Definition~\ref{df:depathing1}, we have $x \in \mks$.
\item If we have $\im(P(x))=\re(P(x)) = 0$, we would have $\im(P(y)) = \im(P(R(x)))=0$, a contradiction.
\item If we have $\im(P(x)) = 0$ and $\re(P(x)) > 0$, $P(x)$ lies on positive half of the real axis, so rotating it clockwise by $2\pi/w$ degrees we would have $\im(P(y))=\im(P(R(x)))<0$, a contradiction.
\end{itemize}
 \cqed \end{proof}}

\subsection{Upper bounding the remaining path length in Mykkeltveit sets}

In this section, we show the remaining path after removing $\mks$ is at most $O(w^3)$ long.
This polynomial bound is a stark contrast to the number of remaining vertices after removing the Mykkeltveit set ---i.e., $\sigma^w - N_{\sigma,w} \sim (1-\frac{1}{w})\sigma^w$, which is exponential in $w$.

Our main argument involves embedding a \wmer to point in the complex plane, similar to Mykkeltveit's construction, while also tracking the weight of the \wmer, which allows us to exploit the intrinsic monotonicity of the embedding.

\CRhide{ 
  \subsubsection{From \wmers to embeddings}
  In this section, we formulate a relaxation that converts paths of \wmers to trajectories in a geometric space.
  Precisely, we model $S_a$ in Lemma~\ref{lm:myk_successor} as a rotation operating on a complex embedding with attached weights, where the weights restrict possible moves.

  Formally, given a pair $(z,t)$ where $z$ is a complex number and $t$ an integer, define the family of operations $Z_{\delta}(z,t) = (r_w^{-1}z+\delta, t+\delta)$.
  When $z=P(x)$ is the position of a \wmer $x$, $t=W(x)=\sum_{i=0}^{w-1}x_i$ is its weight, and when $0 \le \delta + x_0 < \sigma$, $Z_{\delta}(P(x), W(x)) = (P(S_{\delta+x_0}(x)), W(S_{\delta+x_0}(x)))$.
  This means $Z_{\delta}$ is equivalent to finding the position and weight of the successor $S_{\delta+x_0}$.

  We are now looking for the length of the longest path by repeated application of $Z_{\delta}$ that satisfies $0 \le t \le \wmax$, where $\wmax = (\sigma -1)w$ is the maximum weight of any \wmer.
  This is a relaxation of the original problem of finding a longest path as some choices of $\delta$ and some pairs $(z,t)$ on these paths may not correspond to actual transition or \wmer in the \db graph (when $\delta + x_0$ is negative or greater than $\sigma-1$, then it is not a valid transition). 
  In some sense, the pair $(z,t)$ is a loose representation of a \wmer where the precise sequence of the \wmer is ignored and only its weight is considered.
  On the other hand, every valid path in the \db graph corresponds to a path in this relaxation, and an upper-bound on the relaxed problem is an upper-bound of the original problem.

  \subsubsection{Weight-in embedding and relaxation}
  The \emph{weight-in embedding} maps the pair $x=(z,w)$ to the complex plane.
  This transforms the original longest remaining path problem into a geometric problem of bounding the length in the complex plane under some operation $S_{\delta}$.

  \begin{definition}[Weight-In Embedding] \label{df:weightin_embedding}
    The weight-in embedding of $x=(z,t)$ is $Q(x) = z - t$.
    Accordingly, for a \wmer $x$, its embedding is $Q(x) = Q(P(x), W(x)) = P(x)-W(x)$.
  \end{definition}

  The $Z_{\delta}$ operations in this embedding correspond to a rotation, and, maybe surprisingly, this rotation is independent of the value $\delta$.
  \begin{lemma} \label{lm:successor_weightin}
    Let $x=(z,t)$.
    For all $\delta$, the point $Q(Z_{\delta}(x))$ is the point $Q(x)$ rotated clockwise $2\pi/w$ around the point $(-t, 0)$.
  \end{lemma}
  \CRhide{\begin{proof}
      By definition of weight-in embedding and the operation $Z_{\delta}$:
      $$
      Q(Z_{\delta}(z,t)) = r_w^{-1}z + \delta - (t + \delta) = r_w^{-1}(Q(z,t)+t) - t
      $$

      In the complex plane, the rotation formula around center $c$ and of angle $\theta$ is $c + e^{i\theta}(z-c)$.
      Therefore, the operations $Z_{\delta}$ is a rotation around $c=(-t, 0)$ of angle $\theta = -2\pi/w$.
      \cqed \end{proof}}

  Figure~\ref{fig:embedding}(b) shows the weight-in embedding of a \db graph.
  The set $\mathcal{C}_{\sigma, w}=\{(-j, 0) \mid 0\leq j \leq \wmax\}$ is the set of all the possible center of rotations, and is shown by large gray dots on Figure~\ref{fig:embedding}(b).
  Because all the \wmer in a given conjugacy class have the same weight, say $t_0$, the conjugacy classes form a circle around a particular center $(-t_0,0)$.
  The image after application of $S_{\delta}$ is independent of the parameter $\delta$, but dependent on the weight $t$ of the underlying pair $(z,t)$.

  Multiple pairs of $x=(z,t)$ can share the same weight-in embedding $Q(x)$.
  As seen in Figure~\ref{fig:embedding}(b), every node belongs to two circles with different centers, meaning there are two embeddings with same $Q(x)$ but different $t$.

  Lemma~\ref{lm:myk_oneway} naturally divides any path in the \db graph avoiding $\mks$ into two parts, the first part in with $\im(P(x))>0$, and the second part with $\im(P(x))\leq 0$.
  Thanks to the symmetry of the problems, we focus on the upper halfplane, defined as the region with $\im(P(x)) \geq 0$.
  With the weight-in embedding, as long as the path is contained in the upper halfplane, it is always traveling to the right (towards large real value) or stay unmoved, as stated below:

  \begin{lemma}[Monotonicity of $\re(Q(\cdot))$] \label{lm:weightin_monotonicity}
    Assume $Q(x)$ and $Q(Z_{\delta}(x))$ are both in the upper halfplane.
    If $Q(x)$ does not coincide with its associated rotation center $(-t, 0)$, then $\re(Q(Z_{\delta}(x))) > \re(Q(x))$, otherwise $Q(Z_{\delta}(x)) = Q(x)$.
  \end{lemma}
  \CRhide{\begin{proof}
      The operation is a clockwise rotation where the rotation center is on the $x$-axis and the two points are on the non-negative halfplane.
      Necessarily, the real part increased, unless the point is on the fix point of the rotation (which is when $Q(x) = (-t, 0)$).
      \cqed \end{proof}}

  We further relax the problem by allowing rotations from any of the centers in $\mathcal{C}_{\sigma,w}$, not just from some $(-t, 0)$ corresponding to the weight in the weight-in embedding. 
  Lemma~\ref{lm:weightin_monotonicity} still applies in this case and the points in the upper-halfplane move from left to right.
  We are now left with a purely geometric problem involving no \wmers or weights to track:
  \begin{quote}
    What is the longest path $\{z_i\}$ possible where $z_{i+1}$ is obtained from $z_i$ by a rotation of $2\pi/w$ clockwise around a center from $\mathcal{C}_{\sigma,w}$, while staying in the upper halfplane at all times ($\im(z_i) \ge 0, \forall i$)?
  \end{quote}

  We now break the problem into smaller stages as the weight-in embedding pass through rotation centers, defined as $\mathcal{C}_{\sigma, w}=\{(-j, 0) \mid 0\leq j \leq \wmax\}$, the set of points that $Q(x)$ could possibly rotate around regardless of $t$.
  As there are $\wmax+1$ rotation centers and the maximum $\re(Q(x))=\re(P(x))$ for any \wmer is also $\wmax$, we define $2\wmax$ subregions, two between any adjacent pair.
  Formally:

  \begin{definition}[Half Subregions] a subregion is defined as the area $[-j, -j+0.5) \times [0, \wmax]$ called a left subregion or $[-j+0.5, -j+1) \times [0, \wmax]$ called a right subregion, for $0<j\leq\wmax$.
  \end{definition}

  We now define the problem of finding longest path, localized to one left subregion, as follows:

  \begin{definition}[Longest Local Trajectory Problem] \label{df:longest_local_trajectory}
    Define the feasible region $(0, 0.5) \times [0, \wmax]$, and relaxed rotation centers $\mathcal{C}' = \{(j, 0)\mid -\wmax \leq j \leq \wmax\}$. 
    A feasible trajectory is a list of points $\{z_i\}$ such that each point is in the feasible region, and $z_i$ can be obtained by rotating $z_{i-1}$ around $c\in \mathcal{C}'$  clockwise by $2\pi/w$ degrees.
    The solution is the longest feasible trajectory.
  \end{definition}

  Again, note that this new definition is a purely geometric problem involving no \wmers and no weights $W(x)$ to track.
  $z_i$ might stagnate if it coincides with one of the rotation centers, so we do not allow $\re(z_i)=-j$ in this geometric problem.
  Still, it suffices to solve this simpler problem, as indicated by the following lemma:

  \begin{restatable}{lemma}{lmrealrelaxation} \label{lm:longestlocalrelaxation}
    For fixed $w$ and $\sigma$, if the solution to the problem in Definition~\ref{df:longest_local_trajectory} is $L$, the longest path in the \db graph avoiding $\mks$ is upper bounded by $4\wmax L + O(w^2)=O(wL+w^2)$.
  \end{restatable}

  \CR{We prove this in the full version of the paper.}
  {We prove this lemma in Supplementary Section~\ref{supp:tightness}.}

  \subsubsection{Backtracking, heights and local potentials} \label{ssc:full_backtracking}
  In this section, we prove $L=O(w^2)$.
  We will frequently switch between polar and Cartesian coordinates in this section and the next section.
  For simplicity, let $r(z)$ and $\phi(z)$ denote the radius and the polar angle of $z$ written in polar coordinate.
  \begin{lemma} \label{lm:naivetrajectory}
    Any feasible trajectory within the region $(0, d] \times [0, \wmax]$ for $d\leq 0.5$ is at most $O(dw^3)$ long.
  \end{lemma}
  \begin{proof}
    The key observation is if a rotation is not around the origin, $\re(Q(x))$ increases by $\Omega(1/w^2)$.

    To see this, assume $(d, \theta)$ is the polar coordinate of $\re(Q(x))$ with respect to the rotation center.
    The polar coordinate for $\re(Q(S_a(x)))$ is then $(d, \theta - 2\pi/w)$.
    We note that $d\geq 0.5$ as $Q(x)$ satisfies $0<\re(Q(x))<0.5$ and is at least 0.5 away from any other rotation centers.
    The difference in real coordinate is $d(\cos(\theta-2\pi/w)-\cos(\theta))=2d\sin(\theta-\pi/w)\sin(\pi/w)$.
    Now, we require $\theta \in [0, \pi]$ and $\theta-2\pi/w \in [0, \pi]$, so $\sin(\theta-\pi/w) \geq \sin(\pi/w)$ and the whole term is lower bounded by $2d\sin^2(\pi/w) = \Omega(d/w^2)=\Omega(1/w^2)$.

    Only $O(dw^2)$ rotations not around origin is possible in the defined region, otherwise $\re(Q(x))$ would increase by $\Omega(dw^2)\Omega(1/w^2) = \Omega(d)$ already.
    Between two rotations not around the origin, only $w/2$ rotations around the origin can happen, or the point would have rotated $\pi$ degrees and can't stay in the upper halfplane.
    This means the possible number of pure rotations is $O(dw^3)$, which is also the asymptotic upper bound of path length.
    \cqed \end{proof}

  This lemma is sufficient to prove $L=O(w^3)$.
  To obtain $L=O(w^2)$, we need a potential based argument.
  Define $u=1-r_w^{-1}$, and let $s(z)$ be the lowest point above the real axis of form $z+ju$ where $j\in \mathbb{Z}$.
  We can show a potential function of form $E(z)=-wr(s(z))/\pi+\phi(s(z))$ is guaranteed to decrease by at least $2\pi/w$ every rotation, and it can only decrease by $O(w)$ total inside the feasible region, which would complete the proof.
  \CR{We present this proof in the full version of the paper.}{This proof can be found in Supplementary Section~\ref{supp:potential_argument}.}
} 

\subsection{Lower bounding the remaining path length in Mykkeltveit sets}

We provide here a constructive proof of the existence of a $\Omega(w^2)$ long path in the \db graph after removing $\mks$.
Since all \wmers in $\mks$ satisfy $\im(P(x)) \leq 0$, a path satisfying $\im(P(x))>0$ at every step is guaranteed to avoid $\mks$ and our construction will satisfy this criteria.
It suffices to prove the theorem for binary alphabet as the path constructed will also be a valid path in a graph with larger alphabet.
We present the constructions for even $w$ here.


We need an alternative view of \wmers in this section, close to a shift register.
Imagine a paper ring with $w$ slots, labelled tag 0 to tag $w-1$ with content $y=y_0y_1\cdots y_{w-1}$, and a pointer initially at 0.
The \wmer from the ring is $y_jy_{j+1}\cdots y_{w-1}y_0\cdots y_{j-1}=y[j,w-j]\cdot y[0,j]$, assuming pointer is at tag $j$.
A pure rotation $R(x)$ on the ring is simply moving the pointer one base forward, and an impure one $S_a(x)$ is to write $a$ to $y_j$ before moving the pointer forward.

Let $w=2m$.
We create $\lceil w/8 \rceil$ ordered quadruples of tags taken modulo $w$: $Q_j = \{a-j, a+j, b-j, b+j\}$ where $j\in [1, \lceil w/8 \rceil]$, $a=m-1$, and $b=w-1$.
In each quadruple $Q_j$, the set of associated root of unity $r_w^{i+1}$ for the $4$ tags are of form $\{-e^{-i\theta}, -e^{i\theta},e^{-i\theta}, e^{i\theta}\}$, adding up to $0$.
Consequently, changing $y_k$ for each $k$ in $Q_j$ from 1 to 0 does not change the resulting embedding.
The strategy consists of creating ``pseudo-loops'': start from a \wmer, rotate it a certain number of times and switch the bit of the \wmer corresponding to the index in a quadruple to $0$ to return to almost the starting position (the same position in the plane but a different \wmer with lower weight).

More precisely, the initial \wmer $x$ is all ones but $x_{w-1}$ set to zero, with paper ring content $y=x$ and pointer at tag 0.
The resulting \wmer satisfies $P(x) = -1$.
The sequence of operations is as follows.
First, do a pure rotation on $x$.
Then, for each quadruple $Q_j$ from $j=1$ to $j=\lceil w/8 \rceil$, we perform the following actions on $x$: pure rotations until the pointer is at tag $a-j$, impure rotation $S_0$, pure rotations until the pointer is at tag $a+j$, impure rotation $S_0$, pure rotations until pointer is at tag $b-j$, impure $S_0$, pure rotations until pointer is at tag $b+j$, impure $S_0$.

Each round involves exactly $w+1$ rotations since the last step is to an impure rotation $S_0$ at tag $b+j$ which increases by one between quadruple $Q_j$ and $Q_{j+1}$.
The total length of the path over all $Q_i$ is at least $cw^2$ for some constant $c$.
Figure~\ref{fig:longpath} shows an example of quadruples and a generated long path that fits in the upper halfplane.

 \CR{The correctness proof, alongside with the construction for odd $w$ (by approximating movements from $w'=2w$) can be found in the full paper}{The correctness proof for the construction is presented in Supplementary Section~\ref{supp:even_correctness} and the construction for odd $w$ is presented in Supplementary Section~\ref{supp:odd_construction}}.

\begin{figure}
  \centering
  (a)~\raisebox{-.5\height}{\includegraphics{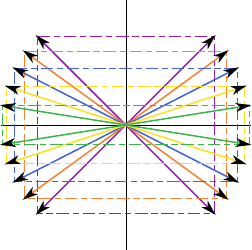}} \hspace{.5in}  (b)~\raisebox{-.5\height}{\includegraphics{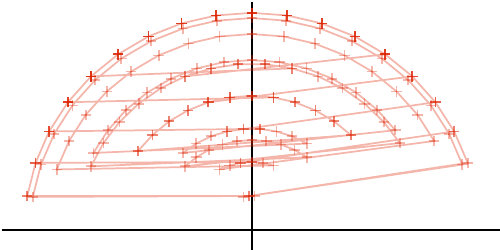}}
  \caption{
    (a)~For $w=40$, each set of $4$ arrows of the same color represent a quadruple set of root of unity.
    There are a total of $5$ sets.
    They were crafted so that the $4$ vector in each set cancel out.
    (b)~The path generated by these quadruple sets.
    The top circle of radius $1$ is traveled many times (between tags $r_1$ and $r_2$ in each quadruple), as after setting the $4$ bits to $0$, the \wmer has the same norm as the starting point.
  }
  \label{fig:longpath}
\end{figure}


\section{Discussion}\label{sec:discussion}

\paragraph{Relationship of UHS and selection schemes.}
Our construction of a \uhs{O(\ln(w)/w)}{w} also implies existence of a forward selection scheme with density $O(\ln(w)/w)$, only a $\ln(w)$ factor away from the lower bound on density achievable by forward scheme and local schemes.

Unfortunately this construction does not apply for arbitrary UHS.
In general, given a UHS with relative size $d$ and remaining path length $w$, it is still unknown how to construct a forward or local scheme with density $O(d)$.
As described in Section~\ref{sec:uhs-from-selection}, we can construct a UHS from a scheme by taking the set $\mathcal{C}_f$ of contexts that yields new selections.
But it is not always possible to go the other way: there are universal hitting sets that are not equal to a set of contexts $\mathcal{C}_f$ for any function $f$.

We are thus interested in the following questions. 
Given a UHS $U$ with relative size $d$, is it possible to create another UHS $U'$ from $U$ that has the same relative size $d$ and correspond to a local scheme (i.e., there exists $f$ such that $U' = \mathcal{C}_f$)?
If not, what is the smallest price to pay (extra density compared to relative size of the UHS) to derive a local scheme from UHS $U$?


\paragraph{Existence of ``perfect'' selection schemes.}
One of the goal in this research is to confirm or deny the existence of asymptotically ``perfect'' selection schemes with density of $1/w$, or at least $O(1/w)$.
Study of UHS might shed light on this problem.
If such perfect selection scheme exists, asymptotic perfect UHS defined as \uhs{O(1/w)}{w} would exist.
On the other hand, if we denied existence of an asymptotic perfect UHS, this would imply nonexistence of ``perfect'' forward selection scheme with density $O(1/w)$.

\paragraph{Remaining path length of Minimum Decycling Sets.}
There is more than one decycling set of minimum size (MDS) for given $w$.
The Mykkeltveit~\cite{mykkeltveit} set is one possible construction, and a construction based on very different ideas is given in Champarnaud \emph{et al.}~\cite{Champarnaud}.
The number of MDS is much larger than the two sets obtained by these two methods.
Empirically, for small values of $w$, we can exhaustively search all the MDS: for $2\le w \le 7$ the number of MDS is respectively \num{2}, \num{4}, \num{30}, \num{28}, \num{68288} and \num{18432}.

While experiments suggest the longest remaining path in a Mykkeltveit depathing set defined in the original paper is around $\Theta(w^3)$, matching our upper bound, we do not know if such bound is tight across all possible minimal decycling sets.
The Champarnaud set seems to have a longer remaining path than the Mykkeltveit set, although it is unknown if it is within a constant factor, bounded by a polynomial of $w$ of different degree, or is exponential.
More generally, we would like to know what is the range of possible remaining path lengths as a function of $w$ over the set of all MDSs.

%
%

\paragraph{Funding:}
This work was partially supported in part by the Gordon and Betty Moore Foundation's Data-Driven Discovery Initiative through Grant GBMF4554 to C.K., by the US National Science Foundation (CCF-1256087, CCF-1319998) and by the US National Institutes of Health (R01GM122935).

\paragraph{Conflict of interests:}
C.K.\ is a co-founder of Ocean Genomics, Inc.
G.M.\ is V.P.\ of software development at Ocean Genomics, Inc.

\bibliographystyle{splncs04}
\bibliography{main}



\newpage
\setcounter{page}{1}
\renewcommand{\thepage}{S\arabic{page}}
\setcounter{table}{0}
\renewcommand{\thetable}{S\arabic{table}}
\setcounter{figure}{0}
\renewcommand{\thefigure}{S\arabic{figure}}
\setcounter{section}{0}
\renewcommand{\thesection}{S\arabic{section}}
\renewcommand{\thetheorem}{S\arabic{theorem}}
\renewcommand{\thedefinition}{S\arabic{definition}}
\renewcommand{\thelemma}{S\arabic{lemma}}
\renewcommand{\thecorollary}{S\arabic{corollary}}

\section{Alternative Proof of Lemma~\ref{lm:uhslocal}}
\label{supp:uhslocal}
\begin{proof}
Assume there exists a path $\{c_0, c_1, \cdots, c_{w-1}\}$ in the \db graph of order $2w-1$ that avoids $\mathcal{C}$.
As $c_{w-1} \notin \mathcal{C}$, there exists some indices $v$ such that $f(c_{w-1}[w-1,w])+(w-1) = f(c_{w-1}[v, w])+v$.
Let $i$ be the smallest $v$ satisfying this property.

As $c_i \notin \mathcal{C}$, there exists some indices $j$ such that $f(c_i[w-1, w]) + (w-1) = f(c_i[j, w])+j$.
Now, note that $c_i[w-1, w] = c_{w-1}[i, w]$, so they have the same value of $f$.
These two terms are underlined below.
\begin{align*}
f(c_{w-1}[w-1, w])+(w-1) & = \underline{f(c_{w-1}[i, w])}+i \\
\underline{f(c_i[w-1, w])} + (w-1) &= f(c_i[j, w])+j \\
2(w-1)-i-j &= f(c_i[j, w]) - f(c_{w-1}[w-1, w]) \leq w-1 \\
i+j-w+1 &\geq 0
\end{align*}

Let $t=i+j-w+1$.
Since $i, j \in [0, w-2]$, $t=i+j-w+1 < w-1$, so it is a valid index between $0$ and $w-2$.
We now note that:
\begin{align*}
f(c_{w-1}[w-1, w])+(w-1) &= f(c_{w-1}[i, w])+i \\
&= f(c_{i}[w-1, w])+(w-1)+i-w+1 \\
&= f(c_i[j, w]) +j + i-w+1 \\
&= f(c_{w-1}[t, w]) + t
\end{align*}
Since $j<w-1$, we have $t<i$, contradicting with the fact that $i$ is the smallest index satisfying $f(c_{w-1}[w-1,w])+(w-1) = f(c_{w-1}[v, w])+v$.
 \end{proof}

\section{Tightness of local trajectory problem}
\label{supp:tightness}
In this section, we prove Lemma~\ref{lm:longestlocalrelaxation} by resolving every difference between Definition~\ref{df:longest_local_trajectory} and the original problem of finding longest path avoiding $\mks$.
We start from the other type of subregions.

\begin{lemma} \label{slm:localrelaxstep1}
For any $0\leq j< \wmax$, any path in the \db graph avoiding $\mks$ has at most $2L+2$ steps satisfying $Q(x) \in (-j-1, -j) \times [0, \wmax]$.
\end{lemma}
\begin{proof}
As the defined region does not contain a rotation center, every move strictly increases $\re(Q(x))$.
We similarly define the left and right subregion as the region with $\re(Q(x)) < -j-0.5$ and $\re(Q(x)) > -j-0.5$.
There is at most one move that goes from the left subregion to the right one, or two moves if there is one \wmer with $\re(Q(x)) = -j-0.5$, and all other moves are contained within either subregion.

For the left subregion, the longest path within it is upper bounded by $L$.
Intuitively, we only need to shift the coordinate to coincide with Definition~\ref{df:longest_local_trajectory}.
Formally, let $\{z_i\}$ be the weight-in embedding of any path strictly within the left subregion.
Then $\{z_{i}+(j+1)\}$  becomes a feasible trajectory under Definition~\ref{df:longest_local_trajectory}, as all points are within the feasible region $(0, 0.5) \times [0, \wmax]$, and each center of rotation, which after shift is $(-W(x)+(j+1), 0) \in \mathcal{C}'$ as $0\leq j< W(x)$.

For the right subregion, we have the same conclusion using a mirroring argument.
To see this, again let $\{z_i\}$ be the weight-in embedding of any path strictly within second subregion.
Then $\{j-\overline{z_{-i}}\}$ ($z_{-i}$ is the $i^\text{th}$ element in $z$ counted backwards, and $\bar{z}$ is the conjugate of $z$) becomes a feasible trajectory.
This is because all points are in the feasible region $(0, 0.5) \times [0, \wmax]$, and assuming $z_{i+1}$ is $z_{i}$ rotated $2\pi/w$ clockwise around $(-t, 0)$, we know $j-\overline{z_{i}}$ is $j-\overline{z_{i+1}}$ rotated $2\pi/w$ clockwise around $(t-j, 0) \in \mathcal{C}'$.
 \cqed \end{proof}

Next, we bound the path length outside any subregions.

\begin{lemma}
Any path in the \db graph satisfying $\re(Q(x)) > 0$ and $\im(Q(x)) \geq 0$ has at most $w/4$ steps.
\end{lemma}
\begin{proof}
We let $\theta$ denote the polar angle of $Q(x)$.
As $Q(x)$ is in first quadrant, $0\leq \theta < \pi/2$.
Next, observe that every rotation around the origin decreases $\theta$ by $2\pi/w$, and every rotation not around origin but some $(-i, 0)$ with $i>0$ will decrease $\theta$ by a greater amount.
This means the path is at most $w/4$ steps long, because in $w/4$ steps $\theta$ would have decreased by at least $\pi/2$, leading to a contradiction.
 \cqed \end{proof}

We can similarly bound the path length left of all regions by looking at the polar angle of $Q(x)$ when origin is at $(-\wmax, 0)$, which leads to the following lemma:

\begin{lemma}
Any path in the \db graph satisfying $\re(Q(x)) < -\wmax$ and $\im(Q(x)) \geq 0$ has at most $w/4$ steps.
\end{lemma}

We are now prove the bound over the entire upper halfplane by bounding path length on the boundaries of subregions.

\begin{lemma} \label{slm:localrelaxstep2a}
Any path in the \db graph satisfying $\im(Q(x)) \geq 0$ has at most $2\wmax L +O(w^2)$ steps.
\end{lemma}

\begin{proof}
We again start by taking $\{z_i\}$ to be the weight-in embedding of any path within the upper halfplane.
We categorize $\{z_i\}$ using their real coordinates.

\begin{itemize}
\item If $\re(z_i) < -\wmax$ or $\re(z_i) > 0$, by previous two lemmas, we know there are at most $w/2$ of them.
\item Else, if $\re(z_i)$ is not an integer, it falls in one of the regions defined by Lemma~\ref{slm:localrelaxstep1}.
As there are $\wmax$ regions total under that definition, and the point can never reenter a region, the point belongs to a path contained within the region of at most $2L+2$ length, and there are at most $\wmax(2L+2)$ points in this category.
\item The last category is when $\re(z_i)$ is an integer.
If $c_i$ satisfies $\im(z_i) > 0$, it will be the only one with this real coordinate as $\re(z_{i+1})>\re(z_i)$.
Otherwise, $c_i$ coincides with one of rotation centers $(-j, 0)$.
It could stay at the same location by doing a rotation around itself, which corresponds to a pure rotation of a \wmer when that \wmer embeds to origin.
By construction of $\mks$ (clause 1 of Definition~\ref{df:depathing1} ), there can only be $w-1$ consecutive moves this way, so at most $w$ elements in $\{z_i\}$ have this real coordinate.
There are $\wmax+1$ possible real coordinates, and each one of them might contain $w$ points, so total number of points in this category is $(\wmax+1)w$.
\end{itemize}

Summing these $3$ categories, we get a bound of $2\wmax L + (\wmax+1)w + w/2 = 2\wmax L + O(w^2)$ for a path in the upper-half plane.
 \cqed \end{proof}

Finally, we look at the path in the lower halfplane with the concept of \wmer complements:

\begin{definition}[Complements of \wmer] 
For $a\in \Sigma$, its complement is defined as $\bar{a}=\sigma-a$.
For a \wmer $x$, its complement $\bar{x}$ is the \wmer $\overline{x_0}\overline{x_1}\cdots \overline{x_{w-1}}$. 
The following property holds:
\begin{itemize}
\item For any $x$, $P(x)=-P(\bar{x})$.
\item For any $x$ and $a$, $P(S_a(x)) = -P(S_{\bar{a}}(\bar{x}))$.
\item If there is an edge $x\rightarrow y$ in the \db graph, there is also an edge $\bar{x} \rightarrow \bar{y}$ in the \db graph.
\end{itemize}
\end{definition}

\begin{lemma} \label{slm:localrelaxstep2b}
Any path in the \db graph satisfying $\im(Q(x)) \leq 0$ has at most $2\wmax L +O(w^2)$ steps.
\end{lemma}

\begin{proof}
For any path satisfying the condition, the path formed by taking complement of every \wmer is a path satisfying $\im(Q(x)) \geq 0$.
By Lemma~\ref{slm:localrelaxstep2a}, the length is also upper bounded by $2\wmax L + O(w^2)$.
 \cqed \end{proof}

We are now ready to prove the original statement as follows:

\lmrealrelaxation*

\begin{proof}
As seen in Lemma~\ref{lm:myk_oneway}, we can bound the path length in two parts.
For the first part with $\im(P(x))>0$, the length is upper bounded by $2\wmax L +O(w^2)$ because we prove a strictly stronger statement in Lemma~\ref{slm:localrelaxstep2a} by also allowing points with $\im(P(x)) = 0$.
For the second part with $\im(P(x))\leq 0$, we also proved a strictly stronger statement in Lemma~\ref{slm:localrelaxstep2b} by also allowing points in $\mks$.
The path length for the original problem is upper bounded by the sum of two upper bounds, which is $4\wmax L + O(w^2)$.
 \cqed \end{proof}

\section{Full Argument for $O(w^3)$ Path Length Upper Bound}
\label{supp:potential_argument}
As mentioned in the main text, Lemma~\ref{lm:naivetrajectory} does not solve our problems because setting $d=0.5$ yields $O(w^3)$ long trajectories.
We can however set $d=10/w$ and now focus on the trajectory in the region $(10/w, 0.5) \times [0, \wmax]$, which we denote as $\mathcal{R}$ for the rest of this section.

We aim to prove $L=O(w^2)$ by proving the near-optimality of a greedy approach: The sequence of rotation such that $z_i$ only rotate around $(1, 0)$ when necessary and otherwise rotate around $(0, 0)$.
Intuitively, if $z_i$ is rotated from further rotation centers, it will move a greater distance towards $\re(z_i) = 0.5$.
For trajectories that include moves that deviate from the greedy trajectory, we want to show we can always backtrack to the move, make corrections and yield a longer trajectory.
We now introduce the tools to formalize this idea.

We focus on the idea of backtracking moves.
Recall the formula $z \leftarrow c + r_w^{-1}(z-c)$ for rotating $z$ around $c$ clockwise by $2\pi/w$ degrees.
Note that this is a linear function of $c$, so if we change $c$ by $(1, 0)$, $z$ will change by $u=1-r_w^{-1}$.
In other words, if $z'$ is the result from rotating $z$ around some centers, to change the rotation centers retroactively, we can simply move $z'$ by a multiple of $u$.
This leads to the following definition.

\begin{definition}[Equivalence Classes and Heights]
Let $u=1-r_w^{-1}$ and recall $\mathcal{R} = (10/w, 0.5) \times [0, \wmax]$.
For any point $z \in \mathcal{R}$, its equivalent set $S(z) = \{z+ju \mid j\in \mathbb{Z}, \im(z+ju) >0\}$.
The point with smallest $j$ in the set is called representative of the set, denoted $s(z)$.
The height of a point is defined as the nonnegative integer $j$ such that $s(z) = z - ju$, which is zero if and only if $z$ is a representative itself.
\end{definition}

Now we can define the potential function.
Loosely speaking, this potential function measures how many steps are left in the trajectory if we strictly follow the greedy approach, backtracking one step if necessary.

\begin{definition}[Local Potential Function]
Let $P(z)=-wr/\pi+\phi$, assuming $z=(r, \phi)$ in polar coordinate.
The potential function of a point is $E(z)=P(s(z))$, where $s(z)$ is representative of $z$.
\end{definition}

For $z\in\mathcal{R}$, it is not guaranteed the representative $s(z)$ is in the same region.
However, we have the following lemma:

\begin{lemma}
If $z\in \mathcal{R}$, $z'$ is obtained by rotating $z$ one step according to the longest trajectory problem, then as long as $z' \in \mathcal{R}$, $s(z') \in \mathcal{R}$.
\end{lemma}
\begin{proof}
By definition of the representative, $\im(s(z')) \geq 0$.
Also by definition of the equivalent set, $s(z')$ is also obtained by rotating $z$ by some points on the real axis, clockwise by $2\pi/w$ degrees.
As we have shown before, such move is guaranteed to increase $\re(z)$, so $\re(s(z')) > \re(z) > 10/w$.
To show $s(z') \in \mathcal{R}$, we only need $\re(s(z')) < 0.5$ and $\im(s(z')) \leq \wmax$.
However, since $\re(s(z')) \leq \re(s(z))$ and $\im(s(z')) \leq \im(s(z))$, $s(z') \notin \mathcal{R}$ would imply $z' \notin \mathcal{R}$.
 \cqed \end{proof}

This means after one rotation in $\mathcal{R}$, or the first step in the trajectory, $s(z)$ is guaranteed to be in the same region and would stay in the region unless $z'$ is already out of the region, indicating end of trajectory.
For the rest of our proofs, we assume $s(z) \in \mathcal{R}$ for the whole trajectory.

Our goal from now on is to prove that $E(z)$ reduces by some amount each rotation.
Assume a rotation brings $z$ to $z'$.
Note that if we only care about $E(z)$, the rotation center is irrelevant as $s(z')$ is the same, so we can assume every move is a rotation around origin and $z'=r_w^{-1}z$, possibly followed by a shift in multiples of $u$ (which does not change $s(z')$).

If $z$ and $z'$ are both of height 0, $s(z)=z$, $s(z')=z'$ and if $z=(r, \phi)$ in polar coordinate, $z'=(r, \phi-2\pi/w)$.
This means $E(z)-E(z') = 2\pi/w$ and the potential drops by $2\pi/w$, a constant value.
If they are both of height $j$, $s(z)=z-ju$, $s(z')=z'-ju$, and it is no longer clear how much $\phi$ changes other than that it decreases a bit.
However, $z'$ is further to the origin, and as we prove below, the change in $r$ is enough for our proofs.
We prove a stronger lemma:

\begin{lemma} \label{slm:potential_general}
If $z'=e^{-i\theta}z$ is $z$ rotated clockwise by $\theta\leq 2\pi/w$ degrees, and they are of the same height $j\geq 1$, then $E(z)-E(z') \geq 1.9\theta$.
\end{lemma}
\begin{proof}
We note that the $z'_0 = z'-ju$ and $z_0 = z-ju$ are the representatives of $z'$ and $z$, and $\phi(z'_0) < \phi(z_0)$.
Let $r'=e^{-i\theta}$.
\begin{align*}
|z'_0| &= |z'-ju| \\
&= |r'z  - ju| \\
&= |r'(z_0+ju) - ju| \\
&= |r'z_0 + j(r'-1)u| \\
&= |z_0 + jr'^{-1}(r'-1)u|
\end{align*}
We let $y = r'^{-1}(r'-1)u$.
Written in polar coordinate, $r'^{-1} = (1, \theta), (r'-1) = (2\sin(\theta/2), 3\pi/2-\theta/2), u=(2\sin(\pi/w), \pi/2-\pi/w)$, so $y=(4\sin(\theta/2)\sin(\pi/w), \theta/2-\pi/w)$.
Since $0<\theta \leq 2\pi/w$, the polar angle of $y$ is between $0$ and $-\pi/w$, which becomes 0 as $w$ grows, meaning $y$ is almost parallel to real axis.

We next bound the polar angle of $z_0$.
Since $z_0$ is in first quadrant, $\phi(z_0) \geq 0$.
Next, since $z_0 \in \mathcal{R}$ and it is a representative, $\re(z_0) \geq 10/w$ and $\im(z_0) \leq \sin(\pi/w)$ (otherwise, $z_0 - t$ still satisfies $\im(\cdot) \geq 0$ and would be the representative instead).
For sufficiently large $w$, we have $\im(z_0) \leq 3/w$ and $\tan(\phi(z_0)) = \im(z_0) / \re(z_0) \leq 0.3$, which yields $\phi(z_0) \leq 0.291$.

Let the angle between $z_0$ and $y$ be $\psi$, we have $|\psi| \leq (\phi(z_0) + 2\pi/w)$.
For large enough $w$, $|\psi| \leq 0.3$ and $\cos(\psi) \geq 0.955$.
Apply the rule of cosines on vector additions:
\begin{align*}
|z'_0|^2 &= |z_0 + jy|^2 \\
&= |z_0|^2 + j^2|y|^2 + 2j\cos(\psi)|z_0||y| \\
& \geq (z_0 + 0.955j|y|)^2
\end{align*}

So $r(z'_0) - r(z_0) \geq 0.955j|y| = 3.82j\sin(\pi/w)\sin(\theta/2) \geq 1.9\pi\theta/w$ for large enough $w$.
We now plug this back to the formula for potential energy:
\begin{align*}
E(z) - E(z') &= P(z_0) - P(z'_0) \\
&= -w(r(z_0) - r(z'_0))/\pi + (\phi(z_0) - \phi(z'_0)) \\
&\geq w(1.9\pi\theta/w)/\pi \\
&= 1.9\theta
\end{align*}

This finishes the proof. 
\cqed \end{proof}

Plugging in $\theta=2\pi/w$, we have the following:

\begin{restatable}{lemma}{lmpotentiala} If $z'=r_w^{-1}z$ with same height $h>0$, $E(z)-E(z') \geq 3.8\pi/w$ for sufficiently large $w$.
\end{restatable}

The last case is when $z$ and $z'$ are of different height.
We need to account for the sudden change of height during rotation.
Intuitively, changing height from $j+1$ to $j$ while making a small movement costs $2\pi/w$ potential, as follows:

\begin{lemma}\label{slm:potential_jump}
For sufficiently large $w$ and a real number $10/w < d < 0.5$, we have $P(d+u) - P(d) = -2\pi/w + O(1/w^2)$.

\end{lemma}
\begin{proof}
We will calculate the difference in $\phi$ and $r$ separately.
Recall that $d=\Omega(1/w)$.
For now, we let $u=a+bi$, that is, $a=2\sin^2(\pi/w),b=2\sin(\pi/w)\cos(\pi/w)=\sin (2\pi/w)$.
For $\phi$, we have:
\begin{align*}
\phi(d+u)-\phi(d) =& \text{arctan}(b/(d+a)) \\
=& \text{arctan}((2\pi/w+O(w^{-3}))/(d+O(w^{-2}))) \\
=& \text{arctan}(2\pi/dw + O(w^{-2})) & d=\Omega(w^{-1})\\ 
=& 2\pi/dw + O(w^{-2}) & \text{arctan}(x) = x+O(x^{-3})
\end{align*}

For $r$, we have:
\begin{align*}
r(d+u) - r(d) =& \sqrt{(d+a)^2+b^2}-d \\
&= \sqrt{d^2+(a+b)^2+2ad}-d \\
&= \sqrt{d^2+4\sin^2(\pi/w)(1+d)}-d \\
&= \frac{4\sin^2(\pi/w) (1+d)}{\sqrt{d^2+4\sin^2(\pi/w)(1+d)}+d} \\
&= \frac{4\pi^2(1+d)/w^2 + O(w^{-4})}{2d+O(w^{-2})} & \sin^2(x)=x^2+O(x^4)=O(x^2)\\
&= \frac{2\pi^2(1+d)}{dw^2}+O(w^{-3}) & d=\Omega(1/w), 1/(d+O(w^{-2}))=1/d+O(w^{-1})
\end{align*}

Merging the two terms we have:
\begin{align*}
P(d+u) - P(d) &= -w(r(d+u)-r(d))/\pi+\phi(d+u)-\phi(d) \\
&= -(2\pi/w) - (2\pi/dw) + (2\pi/dw) + O(w^{-2}) \\
&= -2\pi/w + O(w^{-2})
\end{align*}

This finishes the proof.
 \cqed \end{proof}

Combining previous two lemmas, we can analyze the potential drop for all possible moves and prove the upper bound.

\begin{restatable}{lemma}{lmpotentialb} If $z'=r_w^{-1}z$ with different height, $E(z)-E(z') \geq 3.9\pi/w$ for sufficiently large $w$.
\end{restatable}

\begin{proof}

First of all, the height will only decrease since $z'$ is generated by rotating $z$ around origin in the first quadrant, and $\im(z') < \im(z)$.
For now, assume the height of $z$ is $h$ and height of $z'$ is $h-1$.
We consider the movement of $s(z)$ while rotating from $z$ to $z'$.
There exists one point $z''$ on the arc from $z$ to $z'$ such that from $z$ to $z''$ the height of the point is $h$, and from $z''$ to $z'$ the height is $h-1$.
We can now divide the movement into three parts:
$$
E(z) - E(z') = (E(z) - E(z'')) + (E(z'') - E(z'' + \delta)) + (E(z'' + \delta) - E(z'))
$$
where $\delta$ is an infinitesimal value such that height of $z''$ is $h$ and height of $z''+\delta$ is $h-1$.
The first and last term correspond to the rotation process with constant height.
If the height is 0, the change in potential is exactly the degree rotated.
Otherwise, as shown in Lemma~\ref{slm:potential_general}, the change in potential is at least 1.9 times degree rotated.
Since the total rotated degrees is $2\pi/w$, these two terms add up to at least $2\pi/w$.
The second term corresponds to the change of height as described in Lemma~\ref{slm:potential_jump}, and for sufficiently large $w$, it is at least $1.9\pi/w$.
Adding both terms up, we get $3.9\pi/w$ as desired. 
We can use the same technique if the height drops more than 1 and yield at least the same bounds.
 \cqed \end{proof}

\begin{lemma} $L=O(w^2)$ as in Definition~\ref{df:longest_local_trajectory}.
\end{lemma}
\begin{proof}
We can divide the trajectory into two parts.
The trajectory in the region $(0, 10/w) \times [0, \wmax]$ is $O(w^2)$ long as seen in Lemma~\ref{lm:naivetrajectory}.
The potential function $E(z)$ has a maximum value of $\pi/2$ and minimum value of $O(w)$.
The minimum holds because $r(s(z)) = \sqrt{\re(s(z))^2+\im(s(z))^2}$ is upper bounded by a constant for $\re(s(z)) < 0.5$ and $\im(s(z)) < 1/w$ (otherwise $s(z)-ju$ for $j>0$ would be the representative).
As shown in previous lemmas, each move decreases $E(z)$ by $\Omega(1/w)$, so at most $O(w^2)$ steps are possible in the region $(10/w, 0.5) \times [0, \wmax]$.  \cqed \end{proof}

With Lemma~\ref{lm:longestlocalrelaxation}, we conclude $\mks$ is a UHS with remaining path length $O(w^3)$.

\section{Correctness Proof for even $w$}
\label{supp:even_correctness}

The following lemma implies the created path avoids the Mykkeltveit set.

\begin{lemma} \label{lm:evenbound}
  The sequence generated from the algorithm satisfies $\im(P(x)) > 0$ at every step.
\end{lemma}

\begin{proof}
We define the \emph{absolute embedding} as $P(y)$ on the paper ring model.
This embedding does not change during pure rotations, and if $P(y) = (r, \phi)$ in polar coordinate, the real embedding is $r^{-j}_wP(y)$ or $(r, \phi-2\pi j/w)$ in polar coordinate with pointer at tag $j$.

For each quadruple, we let $v_i = r_w^{r_i+1}$ denote the weight in the absolute embedding for tag $r_i$, where $1\leq i\leq 4$.
We also let $\theta$ be the polar angle of $v_4$.
As seen before, the absolute embeddings are $v_1 = -e^{-i\theta}, v_2 = -e^{i\theta}, v_3 = e^{-i\theta}, v_4 = e^{i\theta}$, and the corresponding polar angles are $\phi(v_1) = \pi-\theta, \phi(v_2) = \pi+\theta, \phi(v_3) = 2\pi-\theta, \phi(v_4) = \theta$.

We now compute the change of $P(y)$, the absolute embedding, and polar angle of $P(x)$ (we use the phrase ``phase" for it) throughout a round as in the following table:

\begin{table}[h]
\centering
\begin{tabular}{|l|l|l|l|l|}
\hline
Stage                      & $P(y)$ Composition      & $P(y)$ Polar Coordinate                   & Starting Phase      & Ending Phase       \\ \hline
Rotation to Tag $r_1$ & $-1$                    & $(1, \pi)$                           & $\pi-\theta$        & $\theta$           \\ \hline
Rotation to Tag $r_2$ & $-1-v_1$                & $(2\sin(\theta/2), 3\pi/2-\theta/2)$ & $(\pi+\theta)/2$    & $(\pi-3\theta)/2$  \\ \hline
Rotation to Tag $r_3$ & $-1-v_1-v_2$            & $(2\cos(\theta)-1, 0)$               & $\pi-\theta$        & $\theta$           \\ \hline
Rotation to Tag $r_4$ & $-1-v_1-v_2-v_3$ & $(2\sin(\theta/2), \pi/2+\theta/2)$  & $(\pi+3\theta)/2$   & $(\pi-\theta)/2$   \\ \hline
End of round               & $-1$                    & $(1, \pi)$                           & \multicolumn{2}{l|}{$\pi-\theta-2\pi/w$} \\ \hline
\end{tabular}
\end{table}

Since $\theta$ increases by $2\pi/w$ each round, the ending condition for one round matches the starting condition for next round.
Before the first round (as we do one pure rotation before first quadruple in the sequence of rotations), the polar angle is at $\pi-2\pi/w$, which matches the starting condition for round 1 with $\theta=2\pi/w$.

As long as $\theta < \pi/3$ (or in our constructions $i<w/6$), during all rotations the polar angle stays between 0 and $\pi$, meaning it stays strictly above the real line and thus avoids $\mks$.
\cqed \end{proof}

\section{Construction for odd $w$}
\label{supp:odd_construction}

We focus on a particular portion of the path from last section with the property that all \wmers in the path are well above the real axis.
We also define the set of \emph{critical embeddings} for a round as the set of embeddings right before or after an impure rotation (or a write in the tape model).
For a quadruple, given $\theta$, the absolute embedding (defined in the proof of Lemma~\ref{lm:evenbound}) and the polar angle of all critical embeddings can be read from the table above.

\begin{lemma} For sufficiently large $w=2m$, the movement sequences defined above from $j=w/20$ to $j=w/10$ satisfies $\im(P(x)) > 0.05$ at every step.
\end{lemma}

\begin{proof} Note that the choice of $j$ means $\theta$ is between $2\pi/10$ and $2\pi/20$.
As pure rotations are arcs over upper halfplane, $\im(P(x))$ is the lowest at the endpoints of pure rotation, or as we defined above, the critical embeddings.
However, at these points, the shortest embedding is $2\sin(\theta/2) > 0.3$, and the smallest polar angle is $\theta$ with $\sin(\theta) > 0.3$.
We then have $\im(P(x)) = r(P(x))\sin(\phi(P(x))) > 0.3 \times 0.3 > 0.05$.
 \cqed \end{proof}

Now let $w=2m+1$ and again assume binary alphabet $\sigma=2$.
Let $a_0 = m-1, a_1 = m, b = w-1 = 2m$. 
The corresponding roots of unity (weights for tags in the absolute embedding) are $r_w^{b+1} = 1$ for $b$, and $r_w^{a_0+1}$ and $r_w^{a_1+1}$ are the roots of unity directly above and below the vector $-1$.
The starting \wmer is full 1 except $a_0, a_1$ and $b$ set to zero. 
The resulting absolute embedding is on the real axis with value $2\cos(\pi/w)-1 = -1 + 4\sin^2(\pi/2w) = -1+O(w^{-2})$.
We now construct a sequence of quadruples such that at the end of every quadruple the absolute embedding is still on the real axis with value close to $-1$.

We let $j$ range from $w/20$ to $w/10$ as described before, but increment it by 2 every step.

\begin{definition}[Imperfect Quadruples]
For each $j$, we construct two candidate quadruples: $Q_j^+ = \{a_0-j, a_1+j, b-j, b+j\}, Q_j^- = \{a_0-j+1, a_1+j-1, b-j, b+j\}$.
Both quadruples satisfy that sum of their corresponding roots of unity is on real axis, which we denote as $q_j^+$ and $q_j^-$.
We have $|q_j| = 2(\cos(\theta+\pi/w)-\cos(\theta)) = O(1/w)$ where $\theta$ is either $2\pi j/w$ or $2\pi j/w + \pi/w$, and are of opposite sign: $q_j^+ >0, q_j^- < 0$.

We define $l_{j}$ to be the embedding by setting all bases in quadruples $\{Q_k \mid k\leq j\}$ to zero from the initial \wmer.
Our construction of the imperfect quadruples ensures $l_j$ is a real number.

We decide the imperfect quadruple to use depending on the sign of $l_{j-2} + 1$.
If $l_{j-2}$ is smaller than $-1$, we pick $Q_j^-$ and we have $l_j = l_{j-2} - q_j^-$.
Otherwise, we pick $Q_j^+$ and we have $l_j = l_{j-2} - q_j^+$.
In both cases, we assured $|l_j +1 | \leq \max(|q_j|, |l_{j-2}+1|)$, which is $O(1/w)$ by induction on $j$.
\end{definition}

The sequence of rotations is defined in exactly the same way as before.
The analyses are similar, as there are between $w+1$ and $w+3$ moves every round and $O(w^2)$ total steps, no two quadruples share tags, and we finish the proof with the following lemma:

\begin{lemma}
For every round of moves using imperfect quadruples, the embedding satisfies $\im(P(x))>0$ at all times.
\end{lemma}

\begin{proof}
Similar to our previous argument, we only need to show $\im(P(x))>0$ at the critical embeddings.
We start by constructing the following (perfect) quadruple for $2w$-mers: $Q'_j = \{w-1-2j, w-1+2j, 2w-1-2j, 2w-1+2j\}$.
As seen in last lemma, the sequence generated by this quadruple satisfies $\im(P(x)) > 0.05$ at all times, so it also holds at the critical embeddings.
We can also map the tags in $Q_j^\pm$ onto $2w$-mers by keeping the corresponding roots of unity the same: $Q_j^{'+}  = \{w-2-2j, w+2j, 2w-1-2j, 2w-1+2j\}$ and $Q_j^{'-}  = \{w-2j, w-2+2j, 2w-1-2j, 2w-1+2j\}$.

Now we fix one embedding in the critical set.
For example, at the end of writing 0 to tag $r_3$, the absolute embedding is $-1-v_1-v_2-v_3=-1+v_4$ and the polar angle is $(\pi+3\theta)/2$ for the perfect quadruple.
We will prove that for the imperfect quadruple, the embedding at this moment is similar.

The absolute embedding is a combination of $v_0$ (which is $-1$ for the perfect quadruple, and $l_{j-2}=-1+O(1/w)$ for the imperfect one) and $\{v_i\}$s.
$v_3$ and $v_4$ are the same for the two quadruples, while $v_1$ and $v_2$ are off by $\pi/w$ degrees, translating to $O(1/w)$ distance on the complex plane.
This means the absolute embedding differs by $O(1/w)$.

The polar angle relative to the absolute embedding is simply one of the $\phi(v_i)$, which is off by at most $\pi/w$.
This corresponds to an extra $\pi/w$ rotation in either direction, and since the length of the embedding is $O(1)$, it moves by $O(1/w)$ on top of the previous argument.

Combining both arguments, we show that if $z$ is the embedding for the perfect quadruple at this moment, the embedding for the imperfect quadruple $z'$ satisfies $|z'-z|=O(1/w)$.
However, $\im(z)>0.05$, so $\im(z')>0$ holds for sufficiently large $w$.
This proof works for all critical embeddings, and since $\im(z')$ is the lowest at critical embeddings, $\im(z')>0$ also holds for the whole round.

 \cqed \end{proof}

\end{document}